\documentclass[lettersize,journal]{IEEEtran}
\usepackage{amsmath,amsfonts,amssymb}
\usepackage{algorithmic}
\usepackage{algorithm}
\usepackage{array}
\usepackage[caption=false,font=normalsize,labelfont=sf,textfont=sf]{subfig}
\usepackage{textcomp}
\usepackage{stfloats}
\usepackage{url}
\usepackage{verbatim}
\usepackage{graphicx}
\usepackage{cite}
\usepackage{booktabs}
\usepackage{multirow}
\usepackage{bm}
\usepackage{mathrsfs}
\usepackage{mathtools}
\usepackage[dvipsnames,table,xcdraw]{xcolor}
\usepackage{colortbl}
\usepackage{makecell}
\usepackage{pifont}
\usepackage{hyperref}
\usepackage{amsthm}

\newtheorem{theorem}{Theorem}
\newtheorem{lemma}[theorem]{Lemma}
\newtheorem{proposition}[theorem]{Proposition}

\theoremstyle{remark}
\newtheorem{remark}{Remark}

\colorlet{phaseI}{rgb:red!2,65;green!30,60;blue!20,125}
\colorlet{phaseII}{rgb:red!2,65;green!30,90;blue!20,125}
\colorlet{phaseIII}{rgb:red!60,100;green!20,90;blue!30,125}

\newcommand{\method}{\textbf{\texttt{PRISM}}}
\newcommand{\best}[1]{\textcolor{red}{\textbf{#1}}}
\newcommand{\second}[1]{\textcolor{blue}{\textbf{#1}}}

\newcommand{\venue}[1]{{\scriptsize\color{gray}[#1]}}
\newcommand{\std}[1]{$_{\pm\text{\scriptsize #1}}$}
\definecolor{deepred}{rgb}{0.6,0,0}
\definecolor{prismblue}{rgb}{0.20,0.40,0.80}

\begin{document}

\title{\method{}: Exposing and Resolving Spurious Isolation in Federated Multimodal Continual Learning}

\author{Beining~Wu,~\IEEEmembership{Member,~IEEE},
        Zihao~Ding,
        and~Jun~Huang,~\IEEEmembership{Senior Member,~IEEE}
\thanks{Beining Wu, Zihao Ding, and Jun Huang are with the Department of Electrical Engineering and Computer Science, South Dakota State University, Brookings, SD 57007, USA. E-mails: \{Wu.Beining, Zihao.Ding\}@jacks.sdstate.edu; Jun.Huang@sdstate.edu.}}

\maketitle

\begin{abstract}
While current federated multimodal continual learning over mixture-of-experts low-rank adaptation (MoE-LoRA) is built on the unverified assumption that routing isolates task-specific knowledge into disjoint experts, we argue that routing operates per-sample, while forgetting accumulates across the task sequence, and gradient conflict persists within each expert even when routing is maximally polarized. Moreover, activation-subspace protection can also fail because, under parameter-efficient fine-tuning, it entangles tasks due to a dimension-counting bound, and federated averaging (FedAvg) disrupts client-side orthogonality. To address this, we propose \method{} (\underline{\textbf{P}}er-expert \underline{\textbf{R}}outing-projection \underline{\textbf{I}}nterference-informed \underline{\textbf{S}}ubspace \underline{\textbf{M}}ethod), which maintains a per-expert gradient subspace basis whose orthogonality is preserved under FedAvg and reinterprets MoE routing as a capacity allocator. Our results show that, on LLaVA-1.5-7B, LLaVA-1.5-13B, and Qwen2.5-VL-7B across CoIN-6 and CoIN-Long-10, \method{} outperforms \emph{sixteen} the state of the art baselines in average accuracy. Compared to the best federated multimodal baseline, the performance margin increases from $+3.23$ pp on CoIN-6 to $+6.06$ pp on CoIN-Long-10.
\end{abstract}



\section{Introduction}
\label{sec:intro}

\IEEEPARstart{M}{ultimodal} large language models are now widely used in federated continual learning, where clients in different locations adapt a shared backbone to their own vision-language tasks without exchanging raw data~\cite{Wu2026TNSE, Wu2026COMST,Wu2025ToN,Ding2025IPCCC,Huang2025TMC,Dong2025TCCN}. In this context, mixture-of-experts low-rank adaptation (MoE-LoRA)~\cite{Chen2024NIPS, Wang2025ICCV, Wei2025NIPS, Meng2026ICLR, Guo2026ICLR} is the main architecture. Here, a router assigns each input to a few experts, which allows the model to increase capacity for each task without retraining the backbone. The main idea is that each task uses different experts, so task-specific knowledge should remain separated in different parameter subsets, even though there is no explicit mechanism to enforce this separation.

Although this idea is widely used, previous work has not shown whether routing that separates inputs in the input space also keeps task knowledge separated in the parameter space. We find that it does not. Even if routing is fixed and made as polarized as possible, gradient conflict still occurs within each expert. This happens because routing works on individual samples, but forgetting builds up over the sequence of tasks. As a result, protection must be applied directly in parameter space. The activation subspace, which methods like GPM~\cite{Saha2021ICLR} and Adam-NSCL~\cite{Wang2021CVPR} protect during full fine-tuning, is a natural choice. However, under parameter-efficient fine-tuning (PEFT), the backbone is frozen, so all tasks share the same representational manifold. This causes the activation subspaces for different tasks to become entangled, so they are no longer task-specific. Even if a task-discriminative target is used, federated aggregation introduces another problem. Federated averaging (FedAvg)~\cite{McMahan2017AISTATS} combines client parameters by taking a weighted mean, but does not preserve orthogonality. Therefore, a protection basis that is orthogonal on each client can lose this property after averaging on the server~\cite{Wu2025WASA,DingICNC2025}.

These issues relate to three main research areas. Multimodal and federated continual learning focuses on reducing forgetting at the task level and handling differences between clients~\cite{Meng2026ICLR, Zhang2025ACL, Ge2025ICML, Guo2025ACL, Yoon2021ICML, Guo2026ICLR}. MoE-LoRA expert routing is used for input-level specialization~\cite{Chen2024NIPS, Dou2024ACL, Wang2025ICCV, Wei2025NIPS, Hou2026ARXIV, Chen2026ARXIV}. Subspace protection methods use projection to reduce forgetting~\cite{Saha2021ICLR, Wang2021CVPR, Wang2023EMNLP, Luo2026ICLR, Qiu2026ICLR,Ding2026ICDCS,Wu2026ARXIV,Fang2025ARXIV,Pudasaini2026HPSR,Wu2023MPE}. However, none of these approaches checks whether routing actually separates tasks in parameter space, considers the entanglement of activation subspaces under PEFT, or ensures that per-expert orthogonality is preserved after federated aggregation. No current method solves all three problems in FMCL.

To address all three problems, we propose \method{}: \underline{\textbf{P}}er-expert \underline{\textbf{R}}outing-projection \underline{\textbf{I}}nterference-informed \underline{\textbf{S}}ubspace \underline{\textbf{M}}ethod. \method{} uses two main mechanisms. First, it keeps a protection basis for each expert in the gradient subspace and aggregates these across clients using a rule that preserves orthogonality under FedAvg. Second, it changes the role of MoE routing. Since parameter-space protection is now in place, routing no longer needs to isolate tasks and instead acts as a way to allocate capacity. An interference-informed scheduler assigns the protection budget to each layer based on the measured conflict, not by tuning a hyperparameter.

In summary, our main contributions are as follows.

\begin{itemize}
    \item We identify Spurious Isolation in MoE-LoRA continual learning (routing polarizes inputs across experts but gradient conflict persists inside each expert) and formalize two causes: structural conflict irreducibility at the task-sequence timescale, and an activation-subspace entanglement bound under PEFT.

    \item We propose \method{}, which maintains per-expert gradient-subspace orthogonality under FedAvg by construction (via Per-Expert Federated Orthogonal Subspace Union, PE-FOSU) and recasts MoE routing as a capacity allocator through routing-projection symbiosis.

    \item We evaluate \method{} on three multimodal backbones (LLaVA-1.5-7B and LLaVA-1.5-13B in the main tables; Qwen2.5-VL-7B in Supp~B) across CoIN-6 and CoIN-Long-10 against \emph{sixteen} baselines; \method{} improves average accuracy and backward transfer over the strongest federated MoE-LoRA and subspace-protection baselines.
\end{itemize}

The rest of this paper is organized as follows. Section~\ref{sec:related} reviews related work on federated multimodal continual learning, MoE-LoRA for continual instruction tuning, and subspace protection. Section~\ref{sec:problem} covers the preliminaries and diagnostic analysis that lead to the five design constraints. Section~\ref{sec:method} describes \method{}. Section~\ref{sec:experiments} presents experiments on three multimodal backbones. Section~\ref{sec:conclusion} gives the conclusion.

\section{Related Work}
\label{sec:related}

\subsection{Federated Multimodal Continual Learning}
\label{subsec:rw_fmcl}

Federated multimodal continual learning (FMCL) couples two research directions that prior work has treated separately. The first, multimodal continual instruction tuning, has progressed through a dedicated benchmark~\cite{Chen2024NIPS}, expansion-based continual tuning~\cite{He2026TIP}, low-rank adaptation (LoRA) rank compression~\cite{Meng2026ICLR}, adapter branching~\cite{Zhang2025ACL}, curriculum expert allocation~\cite{Ge2025ICML}, and hierarchical layer decoupling~\cite{Guo2025ACL}. The second, federated continual learning, addresses client-level non-IID forgetting through sparse parameter decomposition~\cite{Yoon2021ICML}, dual-expert routing~\cite{Guo2026ICLR}, lifecycle-aware forgetting defense~\cite{Wu2026ICDCS}, and a recent formalization of federated continual instruction tuning~\cite{Guo2025ICCV}. Both lines converge on parameter-efficient fine-tuning (PEFT); mixture-of-experts low-rank adaptation (MoE-LoRA) is the prevailing design. Representative methods on each side~\cite{Meng2026ICLR, He2026TIP, Guo2026ICLR, Guo2025ICCV,Pan2023SCIS,Fang2025TON,Fang2025JSAC,Wu2025MNET,Xing2026ACR} address a single axis of the problem; none provides the per-expert protection structure required under federated averaging (FedAvg), and the parameter-space gradient conflict in their composition remains unresolved.

\subsection{Mixture-of-Experts for Continual Instruction Tuning}
\label{subsec:rw_moe}

MoE-LoRA for continual instruction tuning extends PEFT with expert specialization. Task-expert routing has been instantiated through plug-in MoE experts~\cite{Chen2024NIPS, Dou2024ACL}, dual-path specialization against knowledge forgetting~\cite{Wang2025ICCV}, and intra-/inter-modal expert separation~\cite{Wei2025NIPS}. A separate line of recent diagnostics identifies router-expert co-drift~\cite{Hou2026ARXIV} and multi-head routing bottlenecks~\cite{Chen2026ARXIV} as failure modes and proposes routing-space remedies. A common premise underlies all of these designs: with sufficient routing calibration or architectural refinement, task-expert specialization yields knowledge isolation in parameter space. Routing's input-space polarization, however, does not transfer to parameter-space isolation: gradient conflict persists inside each expert even when routing is frozen and maximally polarized.

\subsection{Subspace Protection for Continual Learning}
\label{subsec:rw_subspace}

Subspace protection in continual learning projects updates away from directions that previous tasks have consumed. Methods are organized by which subspace they target: the activation null space~\cite{Saha2021ICLR, Wang2021CVPR}, orthogonal complements of accumulated gradients~\cite{Farajtabar2020AISTATS}, low-rank pretrained weight subspaces~\cite{Xiong2026AAAI}, and gradient-space directions via LoRA-specific projection~\cite{Wang2023EMNLP, Luo2026ICLR, Qiu2026ICLR,Wu2026ARXIV1,Wu2025RACS}. A complementary line establishes a geometric bound~\cite{Steele2026ARXIV} between forgetting and task-gradient principal angles. Yet all of the above are designed for single-adapter, centralized settings. Their design leaves two gaps when the target is federated MoE-LoRA: a single global basis cannot differentiate protection across experts with unequal utilization, so heavily loaded experts go underprotected; and FedAvg does not preserve the orthogonality that client-side projection establishes, so the protection structure degrades after every communication round.

%

\section{Preliminaries and Problem Analysis}
\label{sec:problem}


\subsection{Preliminaries}
\label{subsec:prelim}

\textit{1) Federated Multimodal Continual Learning:}
We study federated multimodal continual learning (FMCL), in which $C$ clients jointly train a shared model on a task sequence $T_1, \ldots, T_N$. Each task $T_t$ has a multimodal instruction-tuning dataset $\mathcal{D}_t = \{(v_i^t, q_i^t, a_i^t)\}_{i=1}^{n_t}$ of visual inputs, text queries, and target answers, partitioned across clients under a Dirichlet$(\beta)$ split~\cite{Hsu2019ARXIV}, where smaller $\beta$ yields more heterogeneous partitions $\mathcal{D}_t^c$. At task $T_t$ the federated objective is the partition-weighted loss $\mathcal{L}_t(\theta) = \sum_{c=1}^{C} (|\mathcal{D}_t^c|/|\mathcal{D}_t|)\, \mathbb{E}_{(v,q,a) \sim \mathcal{D}_t^c}[\ell(\theta; v, q, a)]$, aggregated by federated averaging (FedAvg)~\cite{McMahan2017AISTATS} before $T_{t+1}$ begins. Each new task must be learned with minimal forgetting on $T_1, \ldots, T_{t-1}$, and no raw sample leaves its client.

\textit{2) MoE-LoRA Architecture:}
The shared model adopts a mixture-of-experts low-rank adaptation (MoE-LoRA) architecture~\cite{Dou2024ACL} on top of a frozen multimodal backbone $W$. Each of $E$ experts is a pair of low-rank matrices~\cite{Hu2022ICLR} $A_e \in \mathbb{R}^{r \times d}$ and $B_e \in \mathbb{R}^{d \times r}$ with $r \ll d$. A learned router $W_r$ produces input-dependent weights $\pi_e(x) = \exp((W_r x)_e) / \sum_{e'=1}^{E} \exp((W_r x)_{e'})$, from which the top-$K$ experts form the active set $\mathcal{S}_K(x)$. The forward pass is $y = Wx + \sum_{e \in \mathcal{S}_K(x)} \pi_e(x)\, B_e A_e\, x$, with trainable parameters $\theta = \{A_e, B_e\}_{e=1}^{E} \cup \{W_r\}$. MoE-LoRA has become the predominant architecture for federated multimodal continual learning~\cite{Chen2024NIPS, Wang2025ICCV, Wei2025NIPS, Meng2026ICLR, Guo2026ICLR}, built on the premise that routing assigns different tasks to different experts so that task-specific knowledge stays in disjoint parameter subsets.


\subsection{The Spurious Isolation Phenomenon}
\label{subsec:phenomenon}

We instrument MoE-LoRA with $E=4$, LoRA rank $r=8$, and top-$1$ routing on three multimodal backbones (Qwen2.5-VL-7B, LLaVA-1.5-7B, LLaVA-1.5-13B) trained on CoIN-6~\cite{Chen2024NIPS}.

\begin{figure}[t]
    \centering
    \subfloat[Routing weight.]{%
        \includegraphics[width=0.48\linewidth]{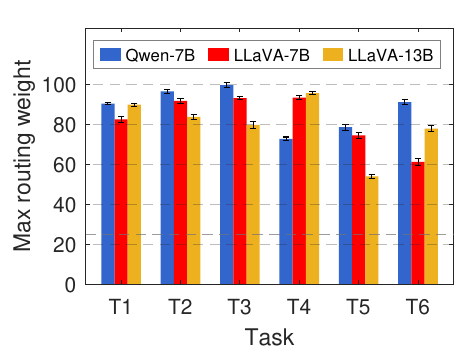}%
        \label{fig:phenomenon-a}}%
    \hfill
    \subfloat[Gradient cosine.]{%
        \includegraphics[width=0.48\linewidth]{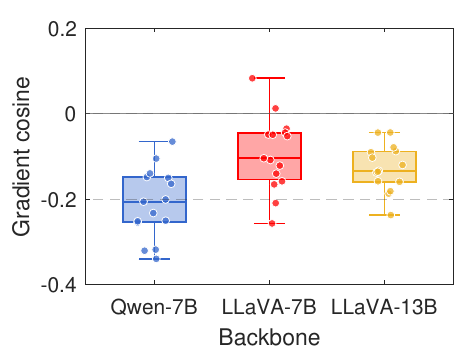}%
        \label{fig:phenomenon-b}}
    \caption{Spurious Isolation evidence on three multimodal backbones.}
    \label{fig:phenomenon}
\end{figure}

\noindent\textbf{Observation 1: Polarized routing coexists with persistent gradient conflict.}
Fig.~\ref{fig:phenomenon}(a) shows each task's dominant expert absorbs $54\%$--$99.5\%$ of its traffic, far above the uniform share $1/E = 25\%$. Yet across all three backbones, $43$ of $45$ cross-task gradient cosines within each pair's most active expert ($\binom{6}{2}=15$ pairs per backbone) fall at or below zero, reaching $-0.34$ on Qwen2.5-VL-7B (Fig.~\ref{fig:phenomenon}b). Polarized routing does not imply parameter-space isolation.

\noindent\textbf{Observation 2: Freezing the router does not eliminate forgetting.}
To rule out routing instability as the cause of Observation 1's conflict, we re-run Qwen2.5-VL-7B with $W_r$ frozen after $T_1$ (routing flip rate $= 0$); backward transfer after $T_6$ nevertheless remains at $-6.8\%$, indistinguishable from the unrestricted baseline (BWT $= -7.1\%$). Routing has no visibility into the per-expert gradient conflict that accumulates in parameter space. During the same diagnostic pass, we record a per-layer interference landscape $\gamma_l$, defined as the rectified cross-task gradient cosine at layer $l$ averaged over task pairs (Tab.~S4 in Supp~B); Section~\ref{sec:method} reads $\gamma_l$ directly without tuning.

Per-expert utilization is heterogeneous across tasks, so a uniform global basis would over-protect lightly loaded experts and under-protect dominant ones; protection must therefore operate at per-expert granularity (design constraint \textbf{C3}). A sample the router effectively ignores contributes negligibly to an expert's parameter update and should influence its protection basis proportionally; the protection covariance must therefore weight each sample by its routing confidence (\textbf{C4}). We label design constraints \textbf{C1}--\textbf{C5} as they arise in the analysis and consolidate them at the end of Section~\ref{subsec:entanglement}.


\subsection{Timescale Mismatch}
\label{subsec:timescale}

Routing operates per sample while forgetting accumulates across the full sequence $T_1, \ldots, T_N$. Let $g_e^{t} = \sum_{i \in \mathcal{D}_t} \pi_e(x_i)\, \nabla_{\theta_e} \ell(\theta; x_i)$ denote the routing-weighted cumulative gradient on expert $e$ across task $T_t$. Proposition~\ref{prop:conflict} shows the resulting structural conflict is irreducible.

\begin{proposition}[Structural Conflict Irreducibility]
\label{prop:conflict}
Let $\mathcal{G}_e^{t} = \{\nabla_{\theta_e} \ell(\theta; x_i) : x_i \in \mathcal{D}_t\}$ be the set of per-sample gradients on expert $e$ for task $T_t$, with cluster mean $\mu_t$. If $\langle \mu_{t_1}, \mu_{t_2} \rangle < 0$ for two tasks $T_{t_1}, T_{t_2}$ and inter-task gradient variance dominates intra-task variance (formal margin condition in Supp.~A, Eq.~(S2)), then $\cos(g_e^{t_1}, g_e^{t_2}) < 0$ for any non-negative routing $\pi$.
\end{proposition}

Because the routing weights are non-negative, each cumulative gradient is confined to the conic hull of its task's per-sample gradients,
\begin{equation}
    g_e^{t} \,\in\, \mathcal{A}_e^{t} \,:=\, \operatorname{cone}\!\big\{\, \nabla_{\theta_e} \ell(\theta; x_i) \,:\, x_i \in \mathcal{D}_t \,\big\}.
    \label{eq:conic}
\end{equation}
When $\langle \mu_{t_1}, \mu_{t_2} \rangle < 0$ and inter-task variance dominates intra-task variance, $\mathcal{A}_e^{t_1}$ and $\mathcal{A}_e^{t_2}$ lie in opposite half-spaces of $\mathbb{R}^d$, and any non-negative combination stays inside its own cone; hence $\cos(g_e^{t_1}, g_e^{t_2})$ remains negative for any choice of $\pi$. Full proof in Supp.~A.

Proposition~\ref{prop:conflict} forces forgetting to be addressed at the task-sequence timescale in parameter space rather than through routing (\textbf{C1}). FedAvg adds a further constraint: cross-client averaging is invisible to routing, and server-side re-orthogonalization would restore the basis only by discarding aggregated gradient mass, so the protection must hold by construction (\textbf{C5}).


\subsection{PEFT Entanglement}
\label{subsec:entanglement}

Constraint C1 leaves the question of which subspace to protect. The continual-learning literature offers three candidates: the activation subspace~\cite{Saha2021ICLR}, the pre-trained weight subspace~\cite{Xiong2026AAAI}, and the gradient subspace~\cite{Luo2026ICLR}. Under parameter-efficient fine-tuning (PEFT) the first two are backbone properties every task inherits rather than task-discriminative signals, and Proposition~\ref{prop:entanglement} makes the resulting entanglement explicit.

\begin{proposition}[Budget-Constrained Activation Entanglement Bound]
\label{prop:entanglement}
Let $r_o$ be the effective rank of the frozen backbone's activation covariance and $k$ the per-expert protection budget retained after singular value decomposition (SVD). For any two tasks $T_{t_1}, T_{t_2}$, the activation-subspace overlap satisfies
\begin{equation}
\Omega^a(t_1, t_2) \;\geq\; \max\!\left(0,\; 2 - \tfrac{r_o}{k}\right).
\label{eq:entanglement}
\end{equation}
\end{proposition}
\begin{proof}[Proof sketch]
Both top-$k$ subspaces lie inside the $r_o$-dimensional manifold; pigeonhole gives $\dim(\cap) \geq 2k - r_o$, equivalent to $\Omega^a \geq 2 - r_o/k$. Full proof in Supp.~A.
\end{proof}

\begin{figure}[t]
    \centering
    \subfloat[Overlap at $k/r_o = 0.75$.]{%
        \includegraphics[width=0.48\linewidth]{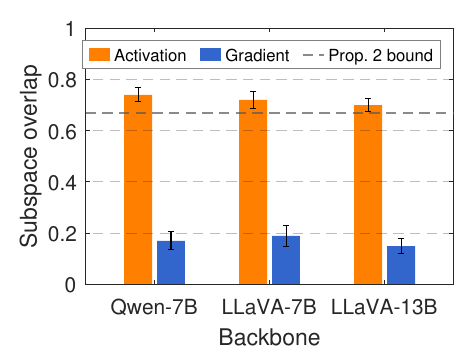}%
        \label{fig:entanglement-a}}%
    \hfill
    \subfloat[Overlap vs. budget.]{%
        \includegraphics[width=0.48\linewidth]{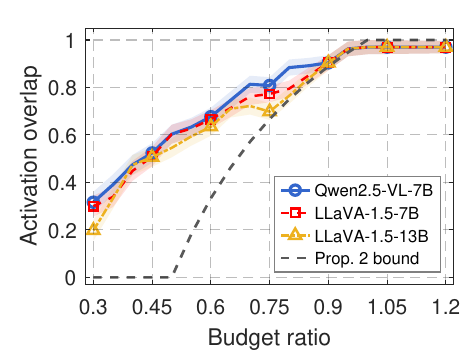}%
        \label{fig:entanglement-b}}
    \caption{Activation versus gradient subspace overlap.}
    \label{fig:entanglement}
\end{figure}

We measure $r_o$ as the number of singular values of the activation covariance above $10^{-3}$ of the leading one. At the operational budget $k/r_o \approx 0.75$, Eq.~\eqref{eq:entanglement} gives $\Omega^a \geq 0.67$. Fig.~\ref{fig:entanglement}(a) measures activation overlap in $[0.70, 0.74]$ against gradient overlap in $[0.15, 0.19]$, a roughly $4\times$ gap; Fig.~\ref{fig:entanglement}(b) shows the bound tightens as the budget ratio grows and closely tracks the measured activation overlap in the practical regime $k/r_o \geq 0.75$, while no such floor binds the gradient subspace.

The reason is structural, visible once the two covariances are written side by side. Let $h_i$ be the activation at sample $i$ and $\delta_i = \partial \ell_i / \partial y_i$ the output-side error signal; then
\begin{equation}
    R \,=\, \sum_{i} h_i\, h_i^{\top}, \qquad G \,=\, \sum_{i} s_i\, h_i\, h_i^{\top}, \quad s_i = \|B^{\top} \delta_i\|_2^{\,2}.
    \label{eq:covariance_contrast}
\end{equation}
The scalar $s_i$ breaks the uniform spectral concentration underlying Eq.~\eqref{eq:entanglement}: $G$'s top eigenvectors concentrate on loss-sensitive samples, $R$'s do not. Activation-subspace protection still preserves prior-task knowledge orthogonally, but high overlap over-allocates capacity to shared directions, collapsing plasticity past the bound (Fed-GPM's AA collapses at $k/r_o > 0.6$ in Fig.~\ref{fig:theory}(b)). The gradient subspace escapes this regime (\textbf{C2}) and enlarges the smallest principal angle $\theta_{\min}$ that drives the forgetting bound $1 - \cos^2 \theta_{\min}$~\cite{Steele2026ARXIV} (Fig.~\ref{fig:theory}(a)).

Combining the timescale analysis with the entanglement bound, any method that resolves Spurious Isolation must satisfy five constraints: parameter-space operation at the task-sequence timescale (\textbf{C1}), a gradient subspace as the protection target (\textbf{C2}), per-expert granularity (\textbf{C3}), routing-weighted accumulation (\textbf{C4}), and FedAvg-compatible orthogonality (\textbf{C5}). Section~\ref{sec:method} presents \method{} as a framework that satisfies all five.

%
%

\section{Methodology}
\label{sec:method}

\begin{figure*}[t]
\centering
\includegraphics[width=0.85\textwidth]{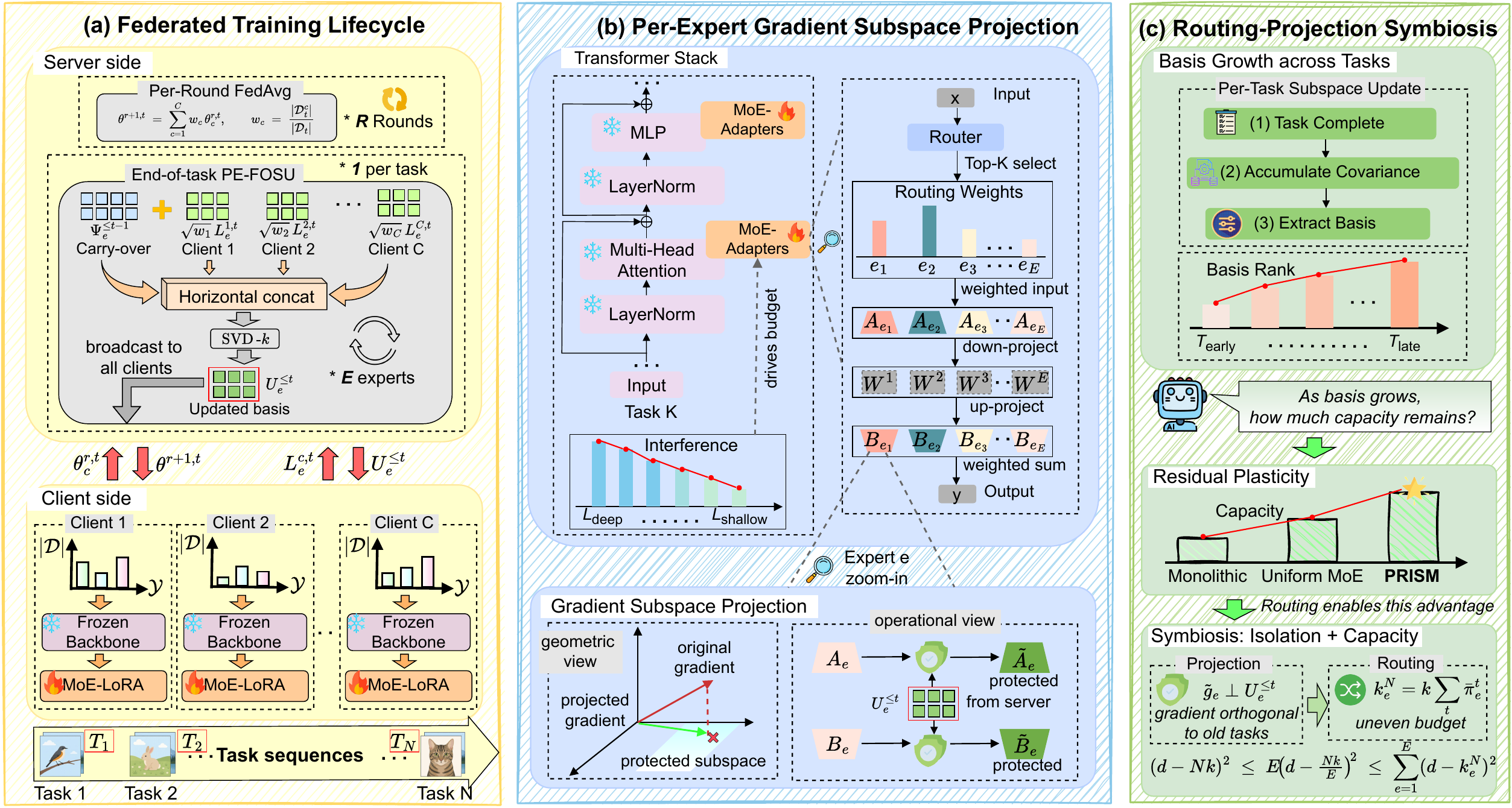}
\caption{Overview of \method{}. (a) Federated lifecycle with per-round FedAvg and end-of-task PE-FOSU. (b) Per-expert gradient subspace projection at deep MoE-LoRA layers, with interference-driven per-layer budget. (c) Routing-projection symbiosis: residual plasticity capacity ordering across architectures.}
\label{fig:framework}
\end{figure*}


\subsection{Overview}
\label{subsec:overview}

Fig.~\ref{fig:framework} traces \method{}'s three coupled mechanisms. Panel~(a) shows the federated training lifecycle: clients train mixture-of-experts low-rank adaptation (MoE-LoRA) on private partitions under per-round federated averaging (FedAvg), and at task end the per-expert federated orthogonal subspace union (PE-FOSU) merges a carry-over operator with weighted client covariance factors through a single thin singular value decomposition (SVD), then broadcasts per-expert bases $U_e^{\leq t}$ to all clients. Panel~(b) zooms into one deep-layer expert: routing weights gate the top-$K$ subset, a Kronecker bilateral projection keeps both LoRA factors orthogonal to $U_e^{\leq t}$, and the per-layer protection rank is read from the interference landscape $\gamma_l$ measured in Section~\ref{subsec:phenomenon}. Panel~(c) closes the design: per-expert protection extends the saturation horizon $E$-fold over single-basis methods on the same total budget, and routing recovers a role as capacity allocator through load matching. The three panels jointly satisfy C1--C5 across Sections~\ref{subsec:pefosu}--\ref{subsec:symbiosis}.

\subsection{Per-Expert Gradient Subspace Isolation with Federated Aggregation}
\label{subsec:pefosu}

To satisfy constraints C1, C2, and C5 jointly, we maintain a per-expert protection basis in the gradient subspace and aggregate it across clients through a rule that preserves orthogonality under FedAvg by construction.

For expert $e$ and client $c$, define the routing-weighted gradient covariance $R_e^{c,t} = \sum_{i \in \mathcal{D}_t^c} \pi_e(x_i)\, g_{i,e}\, g_{i,e}^{\top}$, where $g_{i,e} = \nabla_{\theta_e} \ell(\theta; x_i)$. This definition realizes C4: a sample confidently routed to expert $e$ enters the covariance with large mass, while a sample the router effectively ignored has vanishing influence. Each client factorizes $R_e^{c,t} = L_e^{c,t}(L_e^{c,t})^{\top}$ at the per-layer rank $k = k_l^*$ from the scheduler of Section~\ref{subsec:scheduling} and uploads $L_e^{c,t} \in \mathbb{R}^{d \times k}$.

The server aggregates the client covariances and folds in the carry-over from prior tasks. Writing $\Psi_e^{\leq t-1} = U_e^{\leq t-1}\, \Lambda_e^{\leq t-1}\, (U_e^{\leq t-1})^{\top}$ for the carry-over operator accumulated through $T_{t-1}$, the per-expert federated orthogonal subspace union (PE-FOSU) extracts the top-$k$ eigenpair of the combined covariance:
\begin{equation}
\begin{aligned}
    M_e^{t} \,&=\, \Psi_e^{\leq t-1} \,+\, \sum_{c=1}^{C} w_c\, R_e^{c,t}, \\[4pt]
    (U_e^{\leq t},\, \Lambda_e^{\leq t}) \,&=\, \operatorname*{arg\,max}_{U \,\in\, \mathrm{St}(d,\, k)}\; \operatorname{tr}\!\big(\, U^{\top} M_e^{t}\, U \,\big),
\end{aligned}
\label{eq:pefosu}
\end{equation}
where $w_c = |\mathcal{D}_t^c|/|\mathcal{D}_t|$ matches the FedAvg mixing weights and $\mathrm{St}(d, k) = \{U \in \mathbb{R}^{d \times k} : U^{\top} U = I_k\}$. In implementation, the server performs a thin SVD on the concatenation $[U_e^{\leq t-1} (\Lambda_e^{\leq t-1})^{1/2},\; w_1^{1/2}\, L_e^{1,t},\; \ldots,\; w_C^{1/2}\, L_e^{C,t}] \in \mathbb{R}^{d \times (C+1)k}$, which recovers $(U_e^{\leq t}, \Lambda_e^{\leq t})$ without ever materializing the $d \times d$ matrix $M_e^{t}$; per-expert communication is $O(dk)$ and C1--C3 are met in a single object.

During task $T_{t+1}$, the updated basis governs how clients process incoming gradients. Writing $\Pi_e^{t} = I_d - U_e^{\leq t}(U_e^{\leq t})^{\top}$ for the orthogonal projector onto the complement of $U_e^{\leq t}$, protection on the bilinear LoRA update $\Delta W_e = B_e\, A_e$ is realized by the Kronecker bilateral projector
\begin{equation}
    \mathbf{P}_e^{t} \,=\, \Pi_e^{t}\, \otimes\, \Pi_e^{t}, \qquad \operatorname{vec}\!\big(\widetilde{\Delta W_e}\big) \,=\, \mathbf{P}_e^{t}\, \operatorname{vec}(B_e\, A_e).
    \label{eq:kron_proj}
\end{equation}
The Kronecker product of two orthogonal projectors is itself orthogonal on $\mathbb{R}^{d^2}$, so Eq.~\eqref{eq:kron_proj} is a single parameter-space orthogonal projection rather than the composition of two independent one-sided projections. Un-vectorizing recovers the matrix sandwich $\widetilde{\Delta W_e} = \Pi_e^{t}\, B_e\, A_e\, \Pi_e^{t}$, which at the gradient level decomposes into the pair $g_e^{A} \mapsto g_e^{A}\, \Pi_e^{t}$ and $g_e^{B} \mapsto \Pi_e^{t}\, g_e^{B}$. A one-sided projection leaves a residual $\Delta B_e \cdot A_e\, h$ that is first-order in the update, whereas the bilateral action pushes the residual to second order $\Delta B_e \cdot \Delta A_e\, h$; a formal residual-order argument is given in Supp.~A.

The Kronecker structure fixes the dimensional bookkeeping for Section~\ref{subsec:symbiosis}'s capacity argument.

\begin{lemma}[Bilateral Protection Rank]
\label{lem:bilateral_rank}
The Kronecker bilateral projector $\mathbf{P}_e^{t}$ satisfies
\[
    \operatorname{rank}(\mathbf{P}_e^{t}) = (d - k)^2, \qquad \operatorname{codim}(\mathbf{P}_e^{t}) = 2dk - k^2.
\]
\end{lemma}
\begin{proof}[Proof sketch]
The spectrum of a Kronecker product is the pairwise product of its factor spectra; $\Pi_e^{t}$ has exactly $d - k$ unit eigenvalues, so $\mathbf{P}_e^{t}$ has $(d - k)^2$ unit eigenvalues. The codimension in $\mathbb{R}^{d^2}$ is $d^2 - (d - k)^2 = 2dk - k^2$.
\end{proof}

Lemma~\ref{lem:bilateral_rank} counts directions in the ambient $\mathbb{R}^{d^2}$; the LoRA rank-$r$ constraint refines feasible updates to $\sim 2r(d-k)$ directions per expert, linear in $(d-k)$, which Section~\ref{subsec:symbiosis} uses when comparing per-expert against monolithic capacity.

It remains to verify C5. Client-side projection controls each gradient before aggregation; whether FedAvg preserves that control at the server is settled by the proposition below.

\begin{proposition}[FedAvg Preserves Orthogonality]
\label{prop:fedavg_ortho}
Let $\{\tilde{g}^{c}\}_{c=1}^{C}$ be client gradients that each satisfy $\tilde{g}^{c} \perp U_e^{\leq t}$ on both LoRA factors. For any non-negative weights $\{w_c\}_{c=1}^{C}$ with $\sum_c w_c = 1$, the server-aggregated gradient $\bar{g} = \sum_c w_c\, \tilde{g}^{c}$ also satisfies $\bar{g} \perp U_e^{\leq t}$.
\end{proposition}
\begin{proof}[Proof sketch]
The orthogonal complement of any subspace is closed under finite linear combinations; FedAvg's convex-combination rule falls inside this closure. Full proof in Supp.~A.
\end{proof}

Proposition~\ref{prop:fedavg_ortho} completes C5: orthogonality established at the client survives the federated mixing step, so \method{} achieves FedAvg-compatible protection without server-side correction. Three properties depart from the closest prior design KeepLoRA~\cite{Luo2026ICLR}: per-expert basis (C3), routing-weighted covariance (C4), and FedAvg-exact aggregation (C5).

%

\subsection{Interference-Informed Scheduling}
\label{subsec:scheduling}

The diagnostic setup of Section~\ref{subsec:phenomenon} includes a per-layer interference landscape $\gamma_l \in \mathbb{R}_{\geq 0}$, reported in Tab.~S4 of Supp~B, whose spatial structure is far from uniform; we read protection budgets directly from this measurement rather than tuning them as hyperparameters.

We first allocate protection across layers. Let $L$ be the number of LoRA-inserted layers, $\bar k$ the per-expert protection budget, and $\boldsymbol{\gamma} = (\gamma_1, \ldots, \gamma_L)^{\top}$ the interference vector. The budget is the solution of a water-filling optimization over the simplex,
\begin{equation}
    \{k_l^{*}\}_{l=1}^{L} \,=\, \operatorname*{arg\,max}_{\substack{k_l \,\geq\, 0 \\ \sum_{l} k_l \,=\, \bar k}}\; \sum_{l=1}^{L}\, \gamma_l\, k_l^{\,1/2},
    \label{eq:waterfill}
\end{equation}
whose objective maximizes the total interference energy captured by the allocation. KKT stationarity $\gamma_l / (2\sqrt{k_l^{*}}) = \lambda$ combined with $\sum_l k_l^{*} = \bar k$ yields $k_l^{*} = \bar k\, \gamma_l^{\,2} / \|\boldsymbol{\gamma}\|_2^{\,2}$, so high-conflict layers receive a wider protected subspace and low-conflict layers retain more capacity for plasticity.

We next anneal protection in time. Applying a tight projection from step zero of each new task can starve exploration, so during the warmup window the effective projector interpolates linearly from the identity to $\Pi_e^{t}$: $\Pi_e^{t}(s) = (1 - \alpha(s))\, I_d + \alpha(s)\, \Pi_e^{t}$ with $\alpha(s) = \min(s / s_0,\, 1)$, and $\Pi_e^{t}(s) = \Pi_e^{t}$ once $s \geq s_0$.

%

\subsection{Routing-Projection Symbiosis}
\label{subsec:symbiosis}

Section~\ref{subsec:timescale} proved that routing cannot realize parameter-space isolation; with gradient subspace projection now discharging that role, the remaining question is what routing should do instead. We show that routing recovers a new role: capacity allocation through load matching.

On the rank-$r$ feasible manifold of Section~\ref{subsec:pefosu}, per-expert plasticity capacity scales linearly as $2r(d-k_e)$ in the cumulative per-expert protection rank $k_e$. At horizon $N$ with $\sum_e k_e = Nk$,
\begin{equation}
    \underbrace{2r(d - Nk)}_{\text{monolithic}} \,<\, \underbrace{2r(Ed - Nk)}_{\text{routed (any allocation)}},
    \label{eq:capacity}
\end{equation}
the right-hand side is invariant to the split between uniform and polarized routing because the linear function attains equality under any allocation summing to $Nk$. Read as a saturation-horizon extension over single-basis methods (KeepLoRA~\cite{Luo2026ICLR} and earlier gradient-projection methods): single-basis capacity reaches zero at $Nk = d$, per-expert at $Nk = Ed$, an $E$-fold extension.

The substantive role of polarization is not capacity multiplication but \emph{load matching}. Under polarized routing, $R_e^{c,t}$ concentrates each per-expert basis on its matched tasks' gradient directions; under uniform routing every covariance averages all tasks at $\pi_e \approx 1/E$, and the fixed-$k$ truncation retains only the globally dominant task while minority tasks fall below the spectral cutoff on every expert simultaneously. Fig.~\ref{fig:symbiosis} confirms this: $E_1$ saturates near $0.5$ by task ten while $E_2$--$E_4$ retain residual capacity; the distribution tracks input frequency rather than uniform allocation.

Routing and projection are therefore complementary: projection closes the forgetting gap routing cannot, and routing closes the basis-specificity gap a monolithic projection alone cannot.


\subsection{Training Pipeline}
\label{subsec:pipeline}

Algorithm~\ref{alg:prism} bundles the mechanisms of Sections~\ref{subsec:pefosu}--\ref{subsec:symbiosis} into a single federated lifecycle. Within each task, the per-round loop is a pair of client training and server averaging; at the end of the task, the PE-FOSU subspace union updates the basis that governs projection for the next task.

\begin{algorithm}[t!]
\caption{\method{} Training Lifecycle (\colorbox{rgb:red!2,65;green!30,60;blue!20,125}{Client Local Projection}, \colorbox{rgb:red!30,155;green!20,20;blue!20,30}{Server FedAvg}, and \colorbox{rgb:red!2,65;green!30,90;blue!20,125}{Task-End PE-FOSU})}
\label{alg:prism}
\textbf{Input:} $C$ clients, $N$ tasks, $R$ rounds per task, $E$ experts, per-expert budget $\bar k$, warmup length $s_0$
\begin{algorithmic}[1]
\STATE \textbf{Initialization:} $U_e^{\leq 0} \leftarrow \varnothing$ for $e \in \{1, \ldots, E\}$; read $\gamma_l$ from the diagnostic setup (Section~\ref{subsec:phenomenon}); compute $\{k_l^{*}\}_{l=1}^{L}$ via Eq.~\eqref{eq:waterfill}
\FOR{task $t = 1, \ldots, N$}
    \FOR{round $r = 1, \ldots, R$}
        \colorbox{rgb:red!2,65;green!30,60;blue!20,125}{
        \parbox{0.82\columnwidth}{
        \STATE \textcolor{gray}{$\triangleright$ \textit{Clients perform local training with dual-factor projection:}}
        \FOR{each client $c \in \{1, \ldots, C\}$ \textbf{in parallel}}
            \STATE Receive $\theta^{r,t}$ and basis $\{U_e^{\leq t-1}\}_{e=1}^{E}$
            \FOR{batch $(v, q, a) \sim \mathcal{D}_t^c$ at step $s$}
                \STATE Compute $g_e^{A}, g_e^{B}$ for each $e \in \mathcal{S}_K(x)$
                \STATE Apply Kronecker bilateral projection (Eq.~\eqref{eq:kron_proj}) with warmup weight $\alpha(s)$
                \STATE Update $A_e \leftarrow A_e - \eta\, \tilde{g}_e^{A}$, $B_e \leftarrow B_e - \eta\, \tilde{g}_e^{B}$
                \STATE Accumulate routing-weighted covariance $R_e^{c,t} = \sum_i \pi_e(x_i)\, g_{i,e}\, g_{i,e}^{\top}$
            \ENDFOR
            \STATE Upload $\theta_c^{r,t}$ to the server
        \ENDFOR
        }}

        \colorbox{rgb:red!30,155;green!20,20;blue!20,30}{
        \parbox{0.82\columnwidth}{
        \STATE \textcolor{gray}{$\triangleright$ \textit{Server performs FedAvg on LoRA parameters:}}
        \STATE $\theta^{r+1,t} \leftarrow \sum_{c=1}^{C} w_c\, \theta_c^{r,t}$ with $w_c = |\mathcal{D}_t^c|/|\mathcal{D}_t|$; broadcast
        }}
    \ENDFOR

    \colorbox{rgb:red!2,65;green!30,90;blue!20,125}{
    \parbox{0.82\columnwidth}{
    \STATE \textcolor{gray}{$\triangleright$ \textit{End-of-task PE-FOSU subspace union:}}
    \FOR{each client $c \in \{1, \ldots, C\}$ \textbf{in parallel}}
        \STATE Factorize $R_e^{c,t} = L_e^{c,t}\,(L_e^{c,t})^{\top}$ at rank $k$; upload $L_e^{c,t}$
    \ENDFOR
    \STATE Extract $(U_e^{\leq t}, \Lambda_e^{\leq t})$ via thin SVD on the concatenation of Eq.~\eqref{eq:pefosu}
    \STATE Broadcast $\{U_e^{\leq t}\}_{e=1}^{E}$
    }}
\ENDFOR
\RETURN $\theta^{R,N}$ and $\{U_e^{\leq N}\}_{e=1}^{E}$
\end{algorithmic}
\end{algorithm}

\section{Experiments}
\label{sec:experiments}

\subsection{Experimental Setup}
\label{subsec:setup}

\textit{1) Benchmarks and Backbones:}
We evaluate \method{} on two federated continual multimodal instruction-tuning benchmarks. The first, \textbf{CoIN-6}, is derived from CoIN~\cite{Chen2024NIPS} in the canonical order ScienceQA $\to$ TextVQA $\to$ ImageNet $\to$ GQA $\to$ VQAv2 $\to$ OCR-VQA, covering selection, classification, and open-ended question answering over six tasks. The second, \textbf{CoIN-Long-10}, is a ten-task extension disjoint from CoIN-6 in the order DocVQA $\to$ OKVQA $\to$ ChartQA $\to$ PathVQA $\to$ A-OKVQA $\to$ InfoVQA $\to$ IconQA $\to$ VizWiz $\to$ CLEVR-Math $\to$ AI2D, designed to evaluate long-horizon capacity allocation. Flickr30k and TextCaps are excluded from CoIN-Long-10 because their free-form captioning labels are incompatible with the exact/contains-match evaluation protocol.

\method{} is evaluated on three multimodal backbones with mixture-of-experts low-rank adaptation (MoE-LoRA) adapters injected at the deep layers: LLaVA-1.5-7B (32 layers, L24--31) and LLaVA-1.5-13B (40 layers, L32--39) are reported in the main tables below; Qwen2.5-VL-7B~\cite{Bai2025ARXIV}, the primary diagnostic backbone of Section~\ref{subsec:phenomenon}, is fully evaluated in Supp~B. Each backbone is configured with $E = 4$ experts, rank $r = 8$, scaling $\alpha = 16$, and top-$1$ routing. Each task is partitioned across $C = 5$ clients under a Dirichlet($\beta = 0.3$) split~\cite{Hsu2019ARXIV} and aggregated through federated averaging (FedAvg)~\cite{McMahan2017AISTATS}, following the federated protocol of Fed-Duet~\cite{Guo2026ICLR}. Client-scale and $\beta$ sensitivity analyses appear in Supp~B.

\textit{2) Comparison Methods:}
We compare \method{} against sixteen baselines, grouped as in Table~\ref{tab:main_coin6}. \textit{(i)}~Reference bounds: the zero-shot backbone and centralized multi-task LoRA~\cite{Hu2022ICLR}. \textit{(ii)}~General-purpose federated continual learning baselines: FedProx~\cite{Li2020MLSys}, Fed-EWC~\cite{Kirkpatrick2017PNAS}, Fed-LwF~\cite{Li2016ECCV}, and Fed-Replay~\cite{Rebuffi2017CVPR}. \textit{(iii)}~Federated subspace-protection methods: Fed-GPM~\cite{Saha2021ICLR}, Fed-O-LoRA~\cite{Wang2023EMNLP}, Fed-KeepLoRA~\cite{Luo2026ICLR}, and Fed-SplitLoRA~\cite{Qiu2026ICLR}. \textit{(iv)}~Federated multimodal and MoE-LoRA continual learners: Fed-MoELoRA~\cite{Chen2024NIPS}, Fed-SMoLoRA~\cite{Wang2025ICCV}, Fed-MoDE~\cite{Wei2025NIPS}, Fed-PCLR~\cite{Meng2026ICLR}, Fed-EProj~\cite{He2026TIP}, and Fed-Duet~\cite{Guo2026ICLR}. Non-federated methods are adapted to the federated setting by applying FedAvg to their local update rule with published default hyperparameters. Fed-LwF is omitted on LLaVA-1.5-13B because a frozen teacher copy doubles the 13B memory footprint beyond the per-GPU memory budget.

\textit{3) Evaluation Metrics:}
Let $R_{i,j}$ be the test accuracy on $T_j$ after training on $T_i$, with $N$ the total number of tasks ($N = 6$ or $10$). Following~\cite{LopezPaz2017NIPS}, we report average accuracy $\mathrm{AA} = (1/N)\sum_{j=1}^{N} R_{N,j}$, backward transfer $\mathrm{BWT} = (1/(N{-}1)) \sum_{j=1}^{N-1} (R_{N,j} - R_{j,j})$ (forgetting), and forward transfer $\mathrm{FWT} = (1/(N{-}1)) \sum_{j=2}^{N} R_{j-1,j}$.

\textit{4) Implementation Details:}
All experiments run on four NVIDIA A100 80\,GB GPUs. Following Fed-Duet's federated protocol~\cite{Guo2026ICLR} and the single-epoch convention of multimodal continual instruction tuning~\cite{Chen2024NIPS,Meng2026ICLR}, each task is trained for one FedAvg round with one local epoch per client; this schedule is identical across \method{} and all baselines. Optimization uses AdamW at learning rate $2 \times 10^{-4}$ under cosine decay with a per-client batch size of 16. The Per-Expert Federated Orthogonal Subspace Union (PE-FOSU) budget ratio is fixed at $k/r_o \approx 0.75$ following Fig.~\ref{fig:entanglement}, and the warmup length is $s_0 = 1$ epoch. Reported numbers are averaged over three seeds.

\subsection{Main Results}
\label{subsec:main_results}

\begin{table*}[t]
\centering
\caption{Main results on CoIN-6 at $C=5$. AA and plasticity are in \%; BWT and FWT are percentage points. \best{Red} is the best and \second{blue} the second across competitive methods within each column. \textcolor{gray}{\textit{Italic gray}} rows are reference bounds (zero-shot floor and joint-training oracle), not ranked.}
\label{tab:main_coin6}
\renewcommand{\arraystretch}{1.15}
\setlength{\tabcolsep}{1.5pt}
\footnotesize
\begin{tabular}{c|l|ccc|ccc|ccc}
\toprule
& & \multicolumn{3}{c|}{$\beta=0.1$} & \multicolumn{3}{c|}{$\beta=0.3$} & \multicolumn{3}{c}{$\beta=0.5$} \\
\cmidrule(lr){3-5} \cmidrule(lr){6-8} \cmidrule(lr){9-11}
Backbone & Method & AA & BWT & FWT & AA & BWT & FWT & AA & BWT & FWT \\
\midrule
\multirow{17}{*}{\rotatebox{90}{LLaVA-1.5-7B}} & Zero-shot & {\color{gray}\textit{23.55}\std{0.38}} & {\color{gray}\textit{+0.17}\std{0.26}} & {\color{gray}\textit{-0.02}\std{0.23}} & {\color{gray}\textit{24.05}\std{0.32}} & {\color{gray}\textit{-0.22}\std{0.24}} & {\color{gray}\textit{+0.01}\std{0.18}} & {\color{gray}\textit{23.61}\std{0.31}} & {\color{gray}\textit{+0.14}\std{0.25}} & {\color{gray}\textit{-0.03}\std{0.21}} \\
 & Multi-task & {\color{gray}\textit{71.89}\std{0.29}} & {\color{gray}\textit{-0.02}\std{0.18}} & --- & {\color{gray}\textit{72.18}\std{0.36}} & {\color{gray}\textit{+0.26}\std{0.21}} & --- & {\color{gray}\textit{71.70}\std{0.21}} & {\color{gray}\textit{-0.05}\std{0.26}} & --- \\
\cmidrule(lr){2-11}
 & FedProx~\cite{Li2020MLSys}~\venue{MLSys'20} & 56.12\std{0.46} & -14.42\std{0.98} & -0.11\std{0.43} & 56.61\std{0.44} & -15.09\std{0.65} & -0.22\std{0.35} & 57.40\std{0.47} & -16.06\std{0.76} & -0.36\std{0.33} \\
 & Fed-EWC~\cite{Kirkpatrick2017PNAS}~\venue{PNAS'17} & 51.69\std{0.73} & -7.74\std{1.02} & -1.00\std{0.35} & 52.46\std{0.59} & -8.54\std{0.51} & -0.94\std{0.37} & 52.69\std{0.80} & -8.85\std{0.48} & -1.08\std{0.34} \\
 & Fed-LwF~\cite{Li2016ECCV}~\venue{ECCV'16} & 57.90\std{0.57} & -12.06\std{0.55} & -0.50\std{0.27} & 59.07\std{0.48} & -12.55\std{0.91} & -0.40\std{0.43} & 59.20\std{0.56} & -12.80\std{0.92} & -0.50\std{0.25} \\
 & Fed-Replay~\cite{Rebuffi2017CVPR}~\venue{CVPR'17} & 58.29\std{0.81} & -9.16\std{0.54} & +0.14\std{0.34} & 58.84\std{0.55} & -10.01\std{0.73} & +0.32\std{0.27} & 59.40\std{0.51} & -10.37\std{1.01} & +0.27\std{0.34} \\
\cmidrule(lr){2-11}
 & Fed-GPM~\cite{Saha2021ICLR}~\venue{ICLR'21} & 53.27\std{0.69} & -11.38\std{0.51} & -0.90\std{0.36} & 54.05\std{0.49} & -11.78\std{0.79} & -1.01\std{0.43} & 54.89\std{0.62} & -12.23\std{0.67} & -0.77\std{0.22} \\
 & Fed-O-LoRA~\cite{Wang2023EMNLP}~\venue{EMNLP'23} & 48.22\std{0.72} & -10.72\std{0.76} & -1.39\std{0.41} & 48.82\std{0.81} & -11.31\std{0.80} & -1.36\std{0.43} & 49.34\std{0.51} & -12.18\std{0.97} & -1.35\std{0.38} \\
 & Fed-KeepLoRA~\cite{Luo2026ICLR}~\venue{ICLR'26} & 59.58\std{0.50} & -6.48\std{0.68} & -0.63\std{0.23} & 60.74\std{0.50} & -6.67\std{0.97} & -0.62\std{0.37} & 60.55\std{0.74} & -7.04\std{0.86} & -0.66\std{0.22} \\
 & Fed-SplitLoRA~\cite{Qiu2026ICLR}~\venue{ICLR'26} & 51.19\std{0.79} & -8.51\std{0.70} & -0.72\std{0.24} & 52.24\std{0.57} & -9.29\std{0.79} & -0.68\std{0.26} & 52.48\std{0.68} & -9.89\std{0.69} & -0.84\std{0.43} \\
\cmidrule(lr){2-11}
 & Fed-MoELoRA~\cite{Chen2024NIPS}~\venue{NeurIPS'24} & 55.89\std{0.56} & -8.77\std{0.54} & -0.20\std{0.30} & 56.80\std{0.57} & -9.27\std{1.05} & -0.16\std{0.28} & 57.04\std{0.65} & -9.95\std{0.84} & -0.38\std{0.24} \\
 & Fed-SMoLoRA~\cite{Wang2025ICCV}~\venue{ICCV'25} & 56.38\std{0.47} & -7.80\std{0.82} & -0.13\std{0.37} & 57.66\std{0.51} & -8.33\std{0.83} & -0.21\std{0.43} & 57.62\std{0.79} & -8.90\std{0.70} & -0.21\std{0.42} \\
 & Fed-MoDE~\cite{Wei2025NIPS}~\venue{NeurIPS'25} & 61.64\std{0.49} & -11.87\std{0.93} & +0.01\std{0.31} & 62.37\std{0.48} & -12.09\std{0.64} & +0.39\std{0.36} & 62.78\std{0.55} & -12.64\std{0.81} & \second{+0.37}\std{0.38} \\
 & Fed-PCLR~\cite{Meng2026ICLR}~\venue{ICLR'26} & 51.74\std{0.62} & \best{-3.24}\std{0.63} & -1.42\std{0.27} & 52.55\std{0.60} & \second{-3.98}\std{0.68} & -1.49\std{0.36} & 53.24\std{0.77} & \second{-4.75}\std{0.97} & -1.28\std{0.38} \\
 & Fed-EProj~\cite{He2026TIP}~\venue{TIP'26} & 57.01\std{0.70} & -13.92\std{0.46} & -0.40\std{0.42} & 57.68\std{0.44} & -14.53\std{0.60} & -0.34\std{0.33} & 58.05\std{0.58} & -15.32\std{1.00} & -0.39\std{0.24} \\
 & Fed-Duet~\cite{Guo2026ICLR}~\venue{ICLR'26} & \second{62.07}\std{0.48} & -7.08\std{0.95} & \second{+0.50}\std{0.25} & \second{63.04}\std{0.78} & -7.55\std{0.58} & \second{+0.42}\std{0.42} & \second{62.87}\std{0.74} & -8.14\std{0.69} & +0.35\std{0.23} \\
\cmidrule(lr){2-11}
\rowcolor{prismblue!12}
 & \method{} (Ours) & \best{65.73}\std{0.35} & \second{-3.64}\std{0.47} & \best{+2.32}\std{0.26} & \best{66.27}\std{0.36} & \best{-2.20}\std{0.34} & \best{+2.47}\std{0.17} & \best{66.48}\std{0.33} & \best{-2.16}\std{0.37} & \best{+2.54}\std{0.16} \\
\midrule
\multirow{17}{*}{\rotatebox{90}{LLaVA-1.5-13B}} & Zero-shot & {\color{gray}\textit{26.10}\std{0.33}} & {\color{gray}\textit{+0.59}\std{0.17}} & {\color{gray}\textit{+0.34}\std{0.16}} & {\color{gray}\textit{25.99}\std{0.29}} & {\color{gray}\textit{+0.42}\std{0.20}} & {\color{gray}\textit{+0.16}\std{0.22}} & {\color{gray}\textit{26.06}\std{0.35}} & {\color{gray}\textit{+0.78}\std{0.17}} & {\color{gray}\textit{+0.34}\std{0.21}} \\
 & Multi-task & {\color{gray}\textit{74.99}\std{0.36}} & {\color{gray}\textit{+0.56}\std{0.17}} & --- & {\color{gray}\textit{75.01}\std{0.34}} & {\color{gray}\textit{+0.63}\std{0.19}} & --- & {\color{gray}\textit{74.51}\std{0.26}} & {\color{gray}\textit{+0.45}\std{0.21}} & --- \\
\cmidrule(lr){2-11}
 & FedProx~\cite{Li2020MLSys}~\venue{MLSys'20} & 58.52\std{0.76} & -14.05\std{0.73} & -0.08\std{0.33} & 59.32\std{0.59} & -14.53\std{0.62} & -0.30\std{0.30} & 60.19\std{0.67} & -15.25\std{1.04} & +0.07\std{0.45} \\
 & Fed-EWC~\cite{Kirkpatrick2017PNAS}~\venue{PNAS'17} & 54.29\std{0.61} & -7.40\std{0.84} & -1.07\std{0.22} & 55.23\std{0.63} & -7.78\std{0.83} & -0.71\std{0.37} & 55.34\std{0.62} & -8.45\std{0.83} & -0.91\std{0.35} \\
 & Fed-LwF~\cite{Li2016ECCV}~\venue{ECCV'16} & 60.72\std{0.74} & -10.92\std{0.49} & -0.04\std{0.40} & 61.55\std{0.42} & -11.58\std{0.92} & -0.19\std{0.38} & 61.66\std{0.72} & -12.36\std{0.95} & -0.13\std{0.39} \\
 & Fed-Replay~\cite{Rebuffi2017CVPR}~\venue{CVPR'17} & 60.71\std{0.65} & -8.50\std{0.99} & \second{+0.52}\std{0.34} & 61.44\std{0.55} & -8.90\std{0.62} & +0.53\std{0.42} & 61.97\std{0.49} & -10.04\std{0.74} & +0.36\std{0.32} \\
\cmidrule(lr){2-11}
 & Fed-GPM~\cite{Saha2021ICLR}~\venue{ICLR'21} & 56.12\std{0.61} & -10.81\std{0.69} & -0.71\std{0.22} & 56.60\std{0.70} & -11.02\std{0.98} & -0.60\std{0.37} & 57.09\std{0.38} & -12.08\std{0.53} & -0.65\std{0.25} \\
 & Fed-O-LoRA~\cite{Wang2023EMNLP}~\venue{EMNLP'23} & 51.04\std{0.77} & -10.12\std{0.85} & -1.19\std{0.26} & 51.76\std{0.75} & -10.54\std{0.70} & -1.24\std{0.24} & 51.96\std{0.77} & -11.50\std{0.83} & -0.94\std{0.44} \\
 & Fed-KeepLoRA~\cite{Luo2026ICLR}~\venue{ICLR'26} & 62.62\std{0.40} & -5.53\std{0.69} & -0.42\std{0.37} & 63.14\std{0.82} & -5.80\std{0.77} & -0.44\std{0.37} & 63.66\std{0.54} & -6.59\std{1.02} & -0.35\std{0.41} \\
 & Fed-SplitLoRA~\cite{Qiu2026ICLR}~\venue{ICLR'26} & 53.77\std{0.57} & -7.93\std{0.56} & -0.65\std{0.29} & 54.55\std{0.57} & -8.29\std{1.01} & -0.78\std{0.44} & 55.13\std{0.55} & -9.38\std{0.79} & -0.76\std{0.37} \\
\cmidrule(lr){2-11}
 & Fed-MoELoRA~\cite{Chen2024NIPS}~\venue{NeurIPS'24} & 58.86\std{0.69} & -8.47\std{1.01} & +0.18\std{0.41} & 59.58\std{0.38} & -8.64\std{0.57} & -0.05\std{0.40} & 59.68\std{0.82} & -9.55\std{0.64} & +0.15\std{0.43} \\
 & Fed-SMoLoRA~\cite{Wang2025ICCV}~\venue{ICCV'25} & 59.10\std{0.75} & -7.38\std{0.47} & +0.11\std{0.26} & 59.70\std{0.53} & -7.97\std{0.65} & +0.18\std{0.29} & 60.53\std{0.80} & -8.17\std{0.81} & +0.26\std{0.33} \\
 & Fed-MoDE~\cite{Wei2025NIPS}~\venue{NeurIPS'25} & 64.83\std{0.51} & -10.58\std{0.47} & +0.47\std{0.30} & 65.50\std{0.47} & -11.51\std{0.80} & +0.58\std{0.24} & 65.82\std{0.43} & -12.02\std{0.84} & \second{+0.55}\std{0.21} \\
 & Fed-PCLR~\cite{Meng2026ICLR}~\venue{ICLR'26} & 54.67\std{0.53} & \second{-3.10}\std{0.85} & -1.17\std{0.37} & 55.34\std{0.72} & \second{-3.27}\std{0.65} & -1.01\std{0.39} & 55.55\std{0.42} & \second{-4.10}\std{0.59} & -1.24\std{0.40} \\
 & Fed-EProj~\cite{He2026TIP}~\venue{TIP'26} & 59.55\std{0.66} & -13.36\std{0.55} & -0.33\std{0.36} & 60.56\std{0.49} & -14.33\std{0.82} & -0.41\std{0.26} & 60.74\std{0.42} & -14.65\std{0.84} & -0.27\std{0.31} \\
 & Fed-Duet~\cite{Guo2026ICLR}~\venue{ICLR'26} & \second{64.98}\std{0.81} & -6.81\std{0.93} & +0.45\std{0.38} & \second{65.79}\std{0.48} & -7.18\std{0.46} & \second{+0.61}\std{0.36} & \second{66.15}\std{0.84} & -7.18\std{0.71} & +0.42\std{0.41} \\
\cmidrule(lr){2-11}
\rowcolor{prismblue!12}
 & \method{} (Ours) & \best{69.18}\std{0.34} & \best{-1.70}\std{0.54} & \best{+2.84}\std{0.16} & \best{69.54}\std{0.47} & \best{-1.49}\std{0.48} & \best{+2.74}\std{0.25} & \best{70.10}\std{0.37} & \best{-2.00}\std{0.54} & \best{+2.76}\std{0.22} \\
\bottomrule
\end{tabular}
\end{table*}

\begin{table*}[t]
\centering
\caption{Main results on CoIN-Long-10 (long-sequence stress test) at $C=5$. AA and plasticity are in \%; BWT and FWT are percentage points. \best{Red} is the best and \second{blue} the second across competitive methods within each column. \textcolor{gray}{\textit{Italic gray}} rows are reference bounds (zero-shot floor and joint-training oracle), not ranked.}
\label{tab:main_long10}
\renewcommand{\arraystretch}{1.15}
\setlength{\tabcolsep}{1.5pt}
\footnotesize
\begin{tabular}{c|l|ccc|ccc|ccc}
\toprule
& & \multicolumn{3}{c|}{$\beta=0.1$} & \multicolumn{3}{c|}{$\beta=0.3$} & \multicolumn{3}{c}{$\beta=0.5$} \\
\cmidrule(lr){3-5} \cmidrule(lr){6-8} \cmidrule(lr){9-11}
Backbone & Method & AA & BWT & FWT & AA & BWT & FWT & AA & BWT & FWT \\
\midrule
\multirow{17}{*}{\rotatebox{90}{LLaVA-1.5-7B}} & Zero-shot & {\color{gray}\textit{21.60}\std{0.34}} & {\color{gray}\textit{-0.13}\std{0.21}} & {\color{gray}\textit{-0.02}\std{0.16}} & {\color{gray}\textit{21.56}\std{0.35}} & {\color{gray}\textit{+0.16}\std{0.32}} & {\color{gray}\textit{-0.14}\std{0.23}} & {\color{gray}\textit{21.22}\std{0.31}} & {\color{gray}\textit{+0.06}\std{0.27}} & {\color{gray}\textit{-0.17}\std{0.23}} \\
 & Multi-task & {\color{gray}\textit{52.71}\std{0.31}} & {\color{gray}\textit{-0.23}\std{0.19}} & --- & {\color{gray}\textit{52.72}\std{0.30}} & {\color{gray}\textit{-0.15}\std{0.28}} & --- & {\color{gray}\textit{52.71}\std{0.30}} & {\color{gray}\textit{+0.11}\std{0.18}} & --- \\
\cmidrule(lr){2-11}
 & FedProx~\cite{Li2020MLSys}~\venue{MLSys'20} & 30.33\std{0.46} & -25.83\std{0.74} & -0.85\std{0.29} & 31.04\std{0.52} & -26.25\std{0.87} & -0.64\std{0.41} & 31.44\std{0.54} & -26.86\std{0.63} & -0.85\std{0.32} \\
 & Fed-EWC~\cite{Kirkpatrick2017PNAS}~\venue{PNAS'17} & 26.82\std{0.41} & -13.84\std{0.97} & -1.82\std{0.22} & 28.16\std{0.51} & -14.59\std{0.84} & -2.00\std{0.21} & 28.17\std{0.70} & -14.95\std{1.05} & -1.73\std{0.26} \\
 & Fed-LwF~\cite{Li2016ECCV}~\venue{ECCV'16} & 31.39\std{0.73} & -22.74\std{1.03} & -0.72\std{0.41} & 32.05\std{0.60} & -23.24\std{0.76} & -1.00\std{0.25} & 32.93\std{0.40} & -24.03\std{0.57} & -0.79\std{0.34} \\
 & Fed-Replay~\cite{Rebuffi2017CVPR}~\venue{CVPR'17} & 38.99\std{0.77} & -12.81\std{0.76} & -0.12\std{0.28} & 39.90\std{0.46} & -12.86\std{0.81} & \second{-0.02}\std{0.36} & 40.31\std{0.53} & -13.53\std{0.81} & \second{+0.09}\std{0.30} \\
\cmidrule(lr){2-11}
 & Fed-GPM~\cite{Saha2021ICLR}~\venue{ICLR'21} & 27.37\std{0.81} & -21.41\std{1.04} & -1.58\std{0.34} & 28.48\std{0.54} & -21.92\std{0.80} & -1.55\std{0.25} & 28.72\std{0.53} & -22.21\std{0.98} & -1.52\std{0.20} \\
 & Fed-O-LoRA~\cite{Wang2023EMNLP}~\venue{EMNLP'23} & 20.96\std{0.57} & -21.95\std{0.62} & -2.04\std{0.45} & 22.01\std{0.71} & -22.36\std{0.63} & -2.19\std{0.26} & 22.40\std{0.52} & -23.23\std{0.99} & -1.92\std{0.34} \\
 & Fed-KeepLoRA~\cite{Luo2026ICLR}~\venue{ICLR'26} & 38.29\std{0.71} & -10.26\std{0.99} & -1.10\std{0.33} & 39.13\std{0.52} & -10.64\std{0.85} & -1.17\std{0.33} & 39.30\std{0.62} & -11.16\std{0.75} & -1.27\std{0.29} \\
 & Fed-SplitLoRA~\cite{Qiu2026ICLR}~\venue{ICLR'26} & 24.46\std{0.80} & -17.42\std{0.63} & -1.49\std{0.34} & 24.98\std{0.47} & -18.06\std{0.62} & -1.31\std{0.43} & 25.34\std{0.49} & -18.99\std{0.91} & -1.59\std{0.37} \\
\cmidrule(lr){2-11}
 & Fed-MoELoRA~\cite{Chen2024NIPS}~\venue{NeurIPS'24} & 30.91\std{0.49} & -16.89\std{0.66} & -0.80\std{0.29} & 31.44\std{0.76} & -17.44\std{0.74} & -0.77\std{0.21} & 31.91\std{0.79} & -18.03\std{0.62} & -0.79\std{0.41} \\
 & Fed-SMoLoRA~\cite{Wang2025ICCV}~\venue{ICCV'25} & 33.85\std{0.62} & -11.58\std{0.68} & -0.30\std{0.44} & 34.93\std{0.51} & -12.24\std{0.92} & -0.38\std{0.42} & 35.42\std{0.45} & -12.73\std{0.46} & -0.32\std{0.37} \\
 & Fed-MoDE~\cite{Wei2025NIPS}~\venue{NeurIPS'25} & 37.31\std{0.68} & -21.44\std{0.50} & -0.21\std{0.22} & 37.99\std{0.70} & -22.02\std{1.01} & -0.42\std{0.43} & 38.33\std{0.56} & -22.61\std{0.62} & -0.39\std{0.39} \\
 & Fed-PCLR~\cite{Meng2026ICLR}~\venue{ICLR'26} & 23.77\std{0.79} & \second{-8.77}\std{1.01} & -2.20\std{0.44} & 24.73\std{0.60} & \second{-9.16}\std{0.48} & -2.25\std{0.21} & 25.06\std{0.79} & \second{-9.94}\std{0.72} & -2.19\std{0.24} \\
 & Fed-EProj~\cite{He2026TIP}~\venue{TIP'26} & 31.37\std{0.43} & -24.59\std{0.88} & -0.84\std{0.22} & 32.11\std{0.63} & -25.08\std{0.46} & -1.11\std{0.29} & 32.66\std{0.67} & -26.30\std{0.45} & -0.97\std{0.22} \\
 & Fed-Duet~\cite{Guo2026ICLR}~\venue{ICLR'26} & \second{41.18}\std{0.65} & -11.42\std{0.67} & \second{-0.00}\std{0.32} & \second{41.66}\std{0.84} & -11.54\std{0.85} & -0.15\std{0.37} & \second{42.39}\std{0.62} & -11.66\std{0.50} & -0.10\std{0.21} \\
\cmidrule(lr){2-11}
\rowcolor{prismblue!12}
 & \method{} (Ours) & \best{46.78}\std{0.45} & \best{-2.53}\std{0.32} & \best{+2.07}\std{0.21} & \best{47.72}\std{0.31} & \best{-2.86}\std{0.33} & \best{+1.93}\std{0.23} & \best{47.68}\std{0.33} & \best{-2.81}\std{0.54} & \best{+2.11}\std{0.26} \\
\midrule
\multirow{17}{*}{\rotatebox{90}{LLaVA-1.5-13B}} & Zero-shot & {\color{gray}\textit{23.58}\std{0.35}} & {\color{gray}\textit{+0.78}\std{0.24}} & {\color{gray}\textit{+0.16}\std{0.17}} & {\color{gray}\textit{23.92}\std{0.34}} & {\color{gray}\textit{+0.89}\std{0.18}} & {\color{gray}\textit{+0.20}\std{0.16}} & {\color{gray}\textit{23.49}\std{0.36}} & {\color{gray}\textit{+0.62}\std{0.25}} & {\color{gray}\textit{+0.22}\std{0.24}} \\
 & Multi-task & {\color{gray}\textit{55.34}\std{0.37}} & {\color{gray}\textit{+0.61}\std{0.24}} & --- & {\color{gray}\textit{55.27}\std{0.33}} & {\color{gray}\textit{+0.69}\std{0.18}} & --- & {\color{gray}\textit{55.54}\std{0.25}} & {\color{gray}\textit{+0.73}\std{0.25}} & --- \\
\cmidrule(lr){2-11}
 & FedProx~\cite{Li2020MLSys}~\venue{MLSys'20} & 33.01\std{0.83} & -25.08\std{0.92} & -0.43\std{0.24} & 33.52\std{0.61} & -25.57\std{0.55} & -0.40\std{0.21} & 34.05\std{0.52} & -26.29\std{0.47} & -0.74\std{0.41} \\
 & Fed-EWC~\cite{Kirkpatrick2017PNAS}~\venue{PNAS'17} & 29.75\std{0.66} & -13.26\std{0.45} & -1.59\std{0.27} & 30.62\std{0.40} & -13.97\std{0.71} & -1.67\std{0.20} & 31.08\std{0.67} & -14.79\std{0.47} & -1.75\std{0.28} \\
 & Fed-LwF~\cite{Li2016ECCV}~\venue{ECCV'16} & 34.28\std{0.51} & -22.11\std{0.70} & -0.56\std{0.44} & 34.77\std{0.77} & -22.74\std{0.80} & -0.89\std{0.41} & 35.39\std{0.58} & -23.84\std{0.76} & -0.62\std{0.29} \\
 & Fed-Replay~\cite{Rebuffi2017CVPR}~\venue{CVPR'17} & 41.54\std{0.59} & -11.71\std{0.51} & \second{+0.29}\std{0.30} & 42.67\std{0.41} & -12.54\std{0.83} & +0.03\std{0.40} & 43.20\std{0.81} & -12.91\std{0.52} & \second{+0.28}\std{0.22} \\
\cmidrule(lr){2-11}
 & Fed-GPM~\cite{Saha2021ICLR}~\venue{ICLR'21} & 30.22\std{0.48} & -20.87\std{0.70} & -1.41\std{0.45} & 30.71\std{0.54} & -21.23\std{0.59} & -1.34\std{0.43} & 31.16\std{0.81} & -22.07\std{0.68} & -1.42\std{0.34} \\
 & Fed-O-LoRA~\cite{Wang2023EMNLP}~\venue{EMNLP'23} & 23.67\std{0.51} & -21.05\std{0.79} & -1.81\std{0.30} & 24.80\std{0.54} & -21.61\std{0.74} & -1.96\std{0.26} & 25.13\std{0.58} & -22.73\std{0.91} & -1.96\std{0.29} \\
 & Fed-KeepLoRA~\cite{Luo2026ICLR}~\venue{ICLR'26} & 41.17\std{0.68} & -9.84\std{0.79} & -0.98\std{0.44} & 41.82\std{0.61} & -10.01\std{0.57} & -0.84\std{0.20} & 41.70\std{0.62} & -10.47\std{0.85} & -0.86\std{0.22} \\
 & Fed-SplitLoRA~\cite{Qiu2026ICLR}~\venue{ICLR'26} & 26.85\std{0.73} & -16.96\std{0.56} & -1.31\std{0.38} & 27.54\std{0.76} & -17.69\std{0.99} & -1.26\std{0.35} & 28.24\std{0.61} & -18.38\std{0.74} & -1.24\std{0.33} \\
\cmidrule(lr){2-11}
 & Fed-MoELoRA~\cite{Chen2024NIPS}~\venue{NeurIPS'24} & 33.42\std{0.65} & -16.01\std{0.63} & -0.41\std{0.40} & 34.47\std{0.50} & -16.75\std{0.76} & -0.66\std{0.25} & 34.97\std{0.80} & -17.05\std{0.87} & -0.52\std{0.42} \\
 & Fed-SMoLoRA~\cite{Wang2025ICCV}~\venue{ICCV'25} & 36.87\std{0.80} & -11.05\std{0.73} & -0.42\std{0.34} & 37.64\std{0.50} & -11.80\std{0.46} & -0.21\std{0.25} & 37.86\std{0.72} & -12.12\std{0.94} & -0.27\std{0.21} \\
 & Fed-MoDE~\cite{Wei2025NIPS}~\venue{NeurIPS'25} & 39.96\std{0.62} & -20.37\std{0.81} & +0.01\std{0.21} & 40.81\std{0.49} & -20.97\std{0.92} & -0.22\std{0.38} & 41.47\std{0.77} & -22.14\std{0.82} & -0.23\std{0.28} \\
 & Fed-PCLR~\cite{Meng2026ICLR}~\venue{ICLR'26} & 26.30\std{0.55} & \second{-8.38}\std{0.53} & -1.91\std{0.41} & 27.30\std{0.48} & \second{-8.57}\std{1.00} & -1.99\std{0.39} & 27.47\std{0.71} & \second{-9.46}\std{0.82} & -2.01\std{0.42} \\
 & Fed-EProj~\cite{He2026TIP}~\venue{TIP'26} & 33.83\std{0.83} & -24.04\std{0.67} & -0.85\std{0.27} & 34.86\std{0.44} & -24.53\std{0.92} & -0.87\std{0.23} & 34.85\std{0.63} & -25.65\std{0.89} & -0.62\std{0.21} \\
 & Fed-Duet~\cite{Guo2026ICLR}~\venue{ICLR'26} & \second{43.88}\std{0.82} & -10.30\std{0.68} & +0.05\std{0.24} & \second{44.86}\std{0.52} & -10.93\std{0.89} & \second{+0.10}\std{0.22} & \second{45.22}\std{0.49} & -11.10\std{0.46} & -0.14\std{0.22} \\
\cmidrule(lr){2-11}
\rowcolor{prismblue!12}
 & \method{} (Ours) & \best{50.89}\std{0.46} & \best{-2.15}\std{0.46} & \best{+2.00}\std{0.16} & \best{50.82}\std{0.43} & \best{-2.24}\std{0.51} & \best{+2.01}\std{0.27} & \best{51.44}\std{0.42} & \best{-2.47}\std{0.39} & \best{+2.38}\std{0.20} \\
\bottomrule
\end{tabular}
\end{table*}

Tables~\ref{tab:main_coin6} and~\ref{tab:main_long10} report the main comparison on CoIN-6 and CoIN-Long-10 across sixteen baselines plus \method{} and three Dirichlet concentrations. \method{} holds the top rank on AA in every configuration and on BWT in all but one: Fed-PCLR~\cite{Meng2026ICLR}'s progressive rank compression edges ahead by $0.4$\,pp on LLaVA-1.5-7B at $\beta=0.1$ (the most adverse non-IID setting). The gap inverts on the longer Long-10 sequence, where \method{}'s per-expert capacity outweighs Fed-PCLR's compression advantage. At the headline cell ($\beta=0.3$, $C=5$), \method{} improves BWT by $+4.47$\,pp on CoIN-6 over the monolithic gradient-subspace baseline Fed-KeepLoRA~\cite{Luo2026ICLR}, consistent with the joint effect of C3--C5 absent from the monolithic baseline. Against the pure MoE router Fed-MoDE~\cite{Wei2025NIPS}, AA rises $+3.90$\,pp while forgetting falls $9.9$\,pp, confirming that routing-weighted gradient projection (C1--C4 jointly) outperforms routing-only specialization. Against the strongest federated multimodal baseline Fed-Duet~\cite{Guo2026ICLR}, AA margins widen from $+3.23$\,pp on CoIN-6 to $+6.06$\,pp on Long-10, consistent with orthogonality-preserving aggregation (C5) accumulating larger advantages over longer sequences.

\begin{figure*}[t]
\centering
\includegraphics[width=0.754\textwidth]{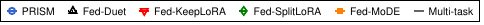}\\[-1.2em]
\subfloat[CoIN-6, LLaVA-7B]{%
    \includegraphics[width=0.245\textwidth]{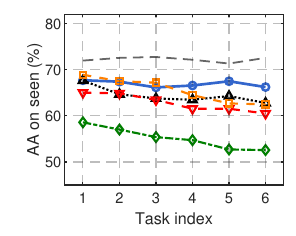}%
    \label{fig:traj:coin6-7b}%
}\hfill
\subfloat[CoIN-6, LLaVA-13B]{%
    \includegraphics[width=0.245\textwidth]{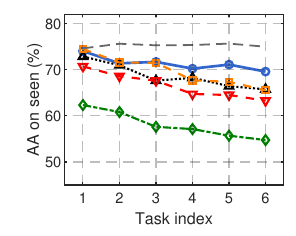}%
    \label{fig:traj:coin6-13b}%
}\hfill
\subfloat[Long-10, LLaVA-7B]{%
    \includegraphics[width=0.245\textwidth]{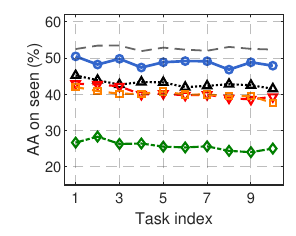}%
    \label{fig:traj:long10-7b}%
}\hfill
\subfloat[Long-10, LLaVA-13B]{%
    \includegraphics[width=0.245\textwidth]{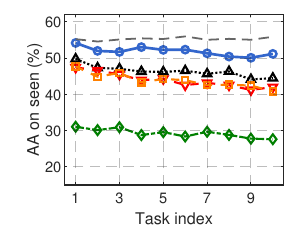}%
    \label{fig:traj:long10-13b}%
}
\caption{Forgetting dynamics across backbones (LLaVA-1.5-7B / 13B) and benchmarks (CoIN-6, CoIN-Long-10). Each panel plots the mean accuracy on all seen tasks at every training stage. \method{} (blue) maintains a near-flat trajectory throughout; the two gradient-subspace baselines (Fed-KeepLoRA, Fed-SplitLoRA) degrade monotonically and steepen on the ten-task sequence as their monolithic basis saturates, and the routing-only and aggregation-only baselines (Fed-MoDE, Fed-Duet) trail further. Multi-task is the centralised upper bound.}
\label{fig:trajectory}
\end{figure*}

Figure~\ref{fig:trajectory} visualizes the forgetting dynamics: \method{}'s trajectory stays nearly flat across all four panels, while every baseline declines monotonically and the decline steepens on the ten-task sequence, with Fed-MoDE and Fed-KeepLoRA falling on opposing sides of the stability--plasticity trade-off.

\subsection{Mechanism and Theoretical Validation}
\label{subsec:mechanism}

\begin{figure*}[t]
\centering
\includegraphics[width=0.88\textwidth]{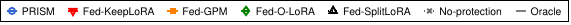}\\[-1.2em]
\subfloat[$\theta_{\min}$ vs task]{%
    \includegraphics[width=0.245\textwidth]{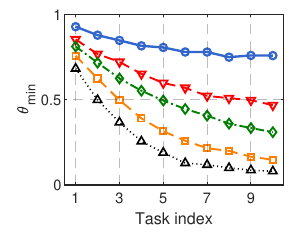}%
    \label{fig:theory:theta}%
}\hspace{0.005\textwidth}
\subfloat[$k/r_o$ sensitivity]{%
    \includegraphics[width=0.245\textwidth]{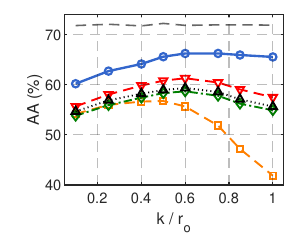}%
    \label{fig:theory:kratio}%
}\hspace{0.005\textwidth}
\subfloat[$\Omega^a$ vs $\Omega^g$ on task pairs]{%
    \includegraphics[width=0.245\textwidth]{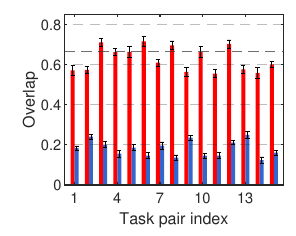}%
    \label{fig:theory:omega}%
}
\caption{Theoretical validation on CoIN-Long-10 (LLaVA-1.5-7B). (a)~$\theta_{\min}$ across tasks: \method{} preserves the projection geometry far better than monolithic gradient-subspace (Fed-KeepLoRA, Fed-SplitLoRA), activation-subspace (Fed-GPM), and no-protection baselines. (b)~Sweeping the budget ratio $k/r_o$: the activation-subspace baseline Fed-GPM collapses beyond $k/r_o=0.6$ as Proposition~2 predicts, while gradient-subspace methods remain stable. \method{} peaks at $k/r_o=0.75$. (c)~Per task pair, the measured activation overlap $\Omega^a$ (red) exceeds the gradient overlap $\Omega^g$ (blue) on all 15 pairs, with $\Omega^a$ respecting the Proposition~2 lower bound (dashed).}
\label{fig:theory}
\end{figure*}

Figure~\ref{fig:theory} validates the three theoretical claims of Section~\ref{sec:method} on real sequences. Panel~(a) tracks $\theta_{\min}$, the smallest principal angle between the accumulated protection basis and the current-task gradient subspace, across ten tasks. \method{} preserves $\theta_{\min} \ge 0.75$ throughout; monolithic gradient-subspace methods degrade to $0.3$--$0.5$, activation-subspace Fed-GPM to $0.15$, and the no-protection baseline collapses to zero, consistent with the forgetting bound in Section~\ref{subsec:entanglement}. Panel~(b) sweeps the budget ratio $k/r_o$: Fed-GPM's AA collapses beyond $k/r_o = 0.6$ as Proposition~2 predicts, while \method{} and Fed-KeepLoRA remain stable on the gradient side, and \method{} peaks at $k/r_o = 0.75$. Panel~(c) reports, per CoIN-6 task pair, the activation overlap $\Omega^a$ (red) and gradient overlap $\Omega^g$ (blue) as paired bars: $\Omega^a > \Omega^g$ on every one of the fifteen pairs, and $\Omega^a$ respects the dashed Proposition~2 lower bound. The three panels back the shift from activation to gradient subspace and from shared to per-expert basis.

\begin{figure}[t]
\centering
\includegraphics[width=\columnwidth]{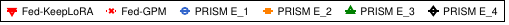}\\[-1.2em]
\subfloat[Null-space, LLaVA-7B]{%
    \includegraphics[width=0.48\linewidth]{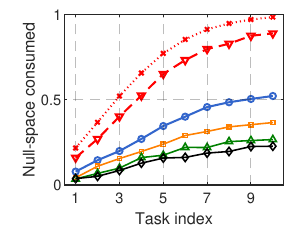}%
    \label{fig:symb:ns7b}%
}\hfill
\subfloat[Null-space, LLaVA-13B]{%
    \includegraphics[width=0.48\linewidth]{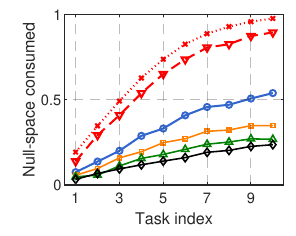}%
    \label{fig:symb:ns13b}%
}
\caption{Routing-projection symbiosis on CoIN-Long-10: per-expert null-space consumption over the 10-task sequence. The activation-subspace baseline Fed-GPM saturates fastest (red dotted) as Proposition~2 predicts, the monolithic gradient-subspace Fed-KeepLoRA follows (red dashed), while \method{}'s four experts share the load heterogeneously and retain capacity, consistent with the load-matching role of polarized routing in Section~\ref{subsec:symbiosis}. Task-expert routing heatmaps after router freeze are deferred to Supp~B.}
\label{fig:symbiosis}
\end{figure}

Figure~\ref{fig:symbiosis} visualizes the routing-projection symbiosis of Section~\ref{subsec:symbiosis}: per-expert null-space consumption over CoIN-Long-10. The activation-subspace Fed-GPM saturates fastest and reaches $0.985$ by task ten as Proposition~2's entanglement bound predicts; monolithic gradient-subspace Fed-KeepLoRA follows at $0.89$; \method{}'s four experts instead absorb the sequence heterogeneously: $E_1$ carries the heaviest load while $E_2$--$E_4$ retain substantial residual capacity. This load distribution embodies the saturation-horizon extension and load matching of Section~\ref{subsec:symbiosis}, Eq.~\eqref{eq:capacity}. Routing-stability metrics confirming the permanence of polarization are deferred to Supp~B.

\subsection{Ablation Study}
\label{subsec:ablation}

\begin{table}[t]
\centering
\caption{Component ablation of \method{} on CoIN-6 at $\beta=0.1$, $C=5$. Each row removes a single mechanism; $\Delta$AA is the drop relative to the full model. AA is in \%; BWT is in percentage points.}
\label{tab:ablation}
\renewcommand{\arraystretch}{1.2}
\setlength{\tabcolsep}{4pt}
\scriptsize
\begin{tabular}{l|cc|cc}
\toprule
& \multicolumn{2}{c|}{LLaVA-1.5-7B} & \multicolumn{2}{c}{LLaVA-1.5-13B} \\
\cmidrule(lr){2-3} \cmidrule(lr){4-5}
Variant & AA & BWT & AA & BWT \\
\midrule
\rowcolor{prismblue!18}
\method{} (full) & \best{65.98\std{0.50}} & \best{-2.25\std{0.42}} & \best{69.18\std{0.43}} & \best{-1.04\std{0.43}} \\
\midrule
\rowcolor{prismblue!5}
w/o PE-FOSU & 61.32\std{0.49} & -6.20\std{0.75} & 64.90\std{0.54} & -5.52\std{0.49} \\
w/o per-expert basis & 62.21\std{0.75} & -5.43\std{0.54} & 65.22\std{0.55} & -4.72\std{0.86} \\
\rowcolor{prismblue!5}
w/o routing-weighted & 62.99\std{0.47} & -4.24\std{0.74} & 66.32\std{0.76} & -3.62\std{0.76} \\
w/o router freeze & 61.79\std{0.81} & -5.22\std{0.50} & 65.40\std{0.63} & -4.54\std{0.71} \\
\rowcolor{prismblue!5}
w/o scheduling & 64.15\std{0.48} & -3.41\std{0.83} & 67.39\std{0.48} & -2.98\std{0.75} \\
w/o warmup ($s_0{=}0$) & 63.11\std{0.83} & -4.44\std{0.96} & 66.23\std{0.64} & -3.64\std{0.94} \\
\rowcolor{prismblue!5}
w/o bilateral (A-only) & 63.37\std{0.45} & -4.15\std{0.63} & 66.33\std{0.46} & -3.32\std{0.99} \\
\bottomrule
\end{tabular}
\end{table}

Table~\ref{tab:ablation} confirms that each design component contributes substantively (numbers below as $7\text{B}/13\text{B}$). Removing PE-FOSU (replacing federated basis aggregation with per-client independent bases) causes the sharpest AA drop ($-4.66/-4.28$\,pp), which confirms C5's necessity under FedAvg. Removing the per-expert basis collapses \method{} to a Fed-KeepLoRA-style monolithic protection and drops AA by $-3.77/-3.96$\,pp with $\Delta$BWT $= -3.18/-3.68$\,pp, in line with the per-expert granularity argument of C3. Unfreezing the router restores routing drift and costs $-4.19/-3.78$\,pp AA with $-2.97/-3.50$\,pp BWT, in line with the routing-stability evidence reported in Supp~B. Interference-informed scheduling contributes the smallest but still non-trivial $-1.83/-1.79$\,pp; its role is to reallocate a fixed budget rather than add new capacity.

Sensitivity to the non-IID concentration $\beta$, client scale $C$, and warmup length $s_0$ is analyzed in Supp~B; across the realistic range of federated configurations, \method{} consistently holds the top position and does not rely on delicate tuning.

%
%

\section{Conclusion}
\label{sec:conclusion}


In federated multimodal continual learning (FMCL) using mixture-of-experts low-rank adaptation (MoE-LoRA), we identified Spurious Isolation: input-space polarization from routing does not guarantee isolation in parameter space. Moreover, when using parameter-efficient fine-tuning (PEFT), activation-subspace protection can cause tasks to become entangled, which we can explain by a dimension-counting argument. To address both of these issues together, we proposed \method{}, in which Per-Expert Federated Orthogonal Subspace Union (PE-FOSU) constructs and maintains a separate gradient subspace basis for each expert that remains orthogonal under federated averaging (FedAvg) by construction, and routing-projection symbiosis reinterprets MoE routing as a capacity allocator. Our experiments on three multimodal backbones, evaluated on CoIN-6 and CoIN-Long-10 and compared to sixteen baselines, show that our approach consistently improves both average accuracy and backward transfer. The improvement widens on the longer ten-task sequence.

\appendices

\renewcommand{\theequation}{S\arabic{equation}}
\renewcommand{\thetheorem}{S\arabic{theorem}}
\renewcommand{\thefigure}{S\arabic{figure}}
\renewcommand{\thetable}{S\arabic{table}}
\setcounter{equation}{0}
\setcounter{theorem}{0}
\setcounter{figure}{0}
\setcounter{table}{0}

\section{Complete Proofs}
\label{supp:proofs}

Notation follows the main text: $g_{i,e} = \nabla_{\theta_e}\ell(\theta; x_i)$, $\pi_e(x_i) \in [0,1]$ with $\sum_e \pi_e(x_i) = 1$, $U_e^{\leq t} \in \mathbb{R}^{d \times k}$ the cumulative protection basis, $P_e^{t} = U_e^{\leq t}(U_e^{\leq t})^{\top}$ and $\Pi_e^{t} = I_d - P_e^{t}$ the orthogonal protection / complement projectors on $\mathbb{R}^d$ with $P_e^{t}\, \Pi_e^{t} = 0$ (distinct from the main-text Kronecker projector $\mathbf{P}_e^{t} = \Pi_e^{t} \otimes \Pi_e^{t}$ on $\mathbb{R}^{d^2}$). Expert and task indices are dropped when unambiguous.

\subsection{Lemma S1: Routing's Operational Constraint}
\label{supp:lemma_routing}

\begin{lemma}[Routing's Operational Constraint]
\label{lem:routing}
Let $\pi : \mathcal{X} \to \Delta^{E-1}$ be a simplex-valued router and let per-sample gradients decompose as $g_{i,e} = \mu_t + \xi_i$ on $\mathcal{D}_t$ with $\|\xi_i\| \le \sigma$. The routing-weighted cumulative gradient $g_e^{t} = \sum_{i \in \mathcal{D}_t} \pi_e(x_i)\, g_{i,e}$ lies in the convex conic hull $\mathcal{A}_e^{t} = \operatorname{cone}\{g_{i,e} : x_i \in \mathcal{D}_t\}$, and its alignment to the cluster mean is bounded by a cone-opening ratio independent of $\pi$:
\begin{equation}
\cos\angle(g_e^{t},\, \mu_t) \,\ge\, \frac{\|\mu_t\| \,-\, \sigma}{\|\mu_t\| \,+\, \sigma} \,=:\, \chi_t.
\label{eq:cone-ceiling}
\end{equation}
\end{lemma}

\begin{proof}
Let $w_t = \sum_{i \in \mathcal{D}_t} \pi_e(x_i) > 0$ and decompose $g_e^{t} = w_t\, \mu_t + \eta_t$ with $\eta_t = \sum_i \pi_e(x_i)\, \xi_i$. Non-negativity of $\pi$ places $g_e^{t}$ in $\mathcal{A}_e^{t}$ by definition, and a triangle inequality with $\|\xi_i\| \le \sigma$ yields $\|\eta_t\| \le w_t\, \sigma$. Bounding the cosine via Cauchy--Schwarz and the triangle inequality,
\begin{align*}
\cos\angle(g_e^{t}, \mu_t)
\,&=\,
\frac{w_t\,\|\mu_t\|^2 \,+\, \langle \eta_t,\, \mu_t\rangle}{\|\, w_t \mu_t + \eta_t \,\|\;\|\mu_t\|} \\
\,&\overset{(a)}{\ge}\,
\frac{w_t\, \|\mu_t\|\,(\|\mu_t\| - \sigma)}{w_t\, (\|\mu_t\| + \sigma)\, \|\mu_t\|}
\,=\,
\frac{\|\mu_t\| - \sigma}{\|\mu_t\| + \sigma},
\end{align*}
where $(a)$ lower-bounds the numerator by $\langle \eta_t, \mu_t\rangle \ge -w_t \sigma\, \|\mu_t\|$ (Cauchy--Schwarz) and upper-bounds the denominator by $\|w_t \mu_t + \eta_t\| \le w_t(\|\mu_t\| + \sigma)$ (triangle inequality).
\end{proof}

\subsection{Proof of Proposition 1: Structural Conflict Irreducibility}
\label{supp:proof_prop1}

Quantify the main-text dominance condition as the margin
\begin{equation}
\langle \mu_1,\, \mu_2\rangle \,\le\, -\delta, \qquad \sigma^2 \,+\, \sigma(\|\mu_1\| + \|\mu_2\|) \,<\, \delta,
\label{eq:dominance}
\end{equation}
with $\delta > 0$ and $\sigma = \max(\sigma_1, \sigma_2)$: the first inequality gives opposition an explicit margin, the second keeps the intra-task radius strictly below it. This formalizes the main-text qualitative condition that inter-task variance dominates intra-task variance: $\delta$ quantifies the inter-task opposition and $\sigma$ bounds the intra-task spread.

\begin{proof}[Proof of Proposition~1]
Using the decomposition $g_e^{t} = w_t\, \mu_t + \eta_t$ with $\|\eta_t\| \le w_t\, \sigma$ from Lemma~\ref{lem:routing},
\begin{align*}
\langle g_e^{t_1},\, g_e^{t_2}\rangle
\,&=\, w_1 w_2\, \langle \mu_1,\, \mu_2\rangle \,+\, w_1\, \langle \mu_1,\, \eta_2\rangle \\
&\quad\,+\, w_2\, \langle \eta_1,\, \mu_2\rangle \,+\, \langle \eta_1,\, \eta_2\rangle \\
\,&\overset{(a)}{\le}\, w_1 w_2 \Big(\, {-\delta} \,+\, \sigma\,(\|\mu_1\| + \|\mu_2\|) \,+\, \sigma^2 \,\Big) \\
\,&\overset{(b)}{<}\, 0,
\end{align*}
where $(a)$ applies Cauchy--Schwarz to each cross term with $\|\eta_t\| \le w_t\, \sigma$ and invokes the opposition inequality in Eq.~\eqref{eq:dominance}, and $(b)$ invokes the dominance inequality in the same equation. Since $w_1 w_2 > 0$ and vector norms are positive, $\cos\angle(g_e^{t_1}, g_e^{t_2}) < 0$ irrespective of $\pi$; no choice of router can flip the numerator without leaving the simplex used in Lemma~\ref{lem:routing}.
\end{proof}

\subsection{Proof of Proposition 2: Budget-Constrained Activation Entanglement Bound}
\label{supp:proof_prop2}

Let $\mathbf{H}_t \in \mathbb{R}^{d \times n_t}$ stack the task-$t$ activations, let $V_t$ be the top-$k$ eigenspace of $R_t = \mathbf{H}_t \mathbf{H}_t^{\top}$ with orthonormal basis $U_t \in \mathbb{R}^{d \times k}$ and projection $P_t = U_t U_t^{\top}$, and let $\mathcal{M} \subset \mathbb{R}^d$ denote the $r_o$-dimensional dominant manifold of the frozen backbone so that $V_t \subset \mathcal{M}$ for every $t$.

\begin{proof}[Proof of Proposition~2]
The activation overlap chains through the Frobenius--trace identity and SVD, the Grassmann dimension formula, and the manifold inclusion:
\begin{align*}
\Omega^a(t_1, t_2)
\,&\overset{(a)}{=}\,
\frac{\operatorname{tr}\!\big( P_{t_1}\, P_{t_2} \big)}{k}
\,=\,
\frac{1}{k} \sum_{j=1}^{k} \cos^2 \theta_j \\
\,&\overset{(b)}{\ge}\,
\frac{\dim(\, V_{t_1} \cap V_{t_2} \,)}{k} \\
\,&\overset{(c)}{=}\,
\frac{\dim V_{t_1} \,+\, \dim V_{t_2} \,-\, \dim(\, V_{t_1} + V_{t_2} \,)}{k} \\
\,&\overset{(d)}{\ge}\,
\frac{2k - r_o}{k}
\,=\,
2 \,-\, \frac{r_o}{k},
\end{align*}
where $(a)$ applies the Frobenius identity $\|U_{t_1}^{\top} U_{t_2}\|_F^{\,2} = \operatorname{tr}(P_{t_1} P_{t_2})$ and the SVD of $U_{t_1}^{\top} U_{t_2}$ with principal angles $\{\theta_j\}_{j=1}^{k}$, $(b)$ retains only the $\theta_j = 0$ contributions that span $V_{t_1} \cap V_{t_2}$, $(c)$ is the Grassmann dimension formula, and $(d)$ uses $V_{t_1} + V_{t_2} \subset \mathcal{M}$ with $\dim\mathcal{M} = r_o$. The regime $k > r_o/2$ forces strictly positive overlap; equality is attained once $V_{t_1} + V_{t_2}$ exhausts $\mathcal{M}$.
\end{proof}

\subsection{Bilateral Projection Residual-Order Argument}
\label{supp:bilateral_residual}

Let $P$, $\Pi = I - P$ be the protection / complement projectors on $\mathbb{R}^d$. A gradient step of size $\eta$ maps $(A_e, B_e) \mapsto (A_e^{+}, B_e^{+})$ via
\begin{align}
\text{one-sided:} \;\; \Delta A_e &\,=\, -\eta\, g^{A}\, \Pi, & \Delta B_e &\,=\, -\eta\, g^{B}, \label{eq:one-sided} \\
\text{bilateral:} \;\; \Delta A_e &\,=\, -\eta\, g^{A}\, \Pi, & \Delta B_e &\,=\, -\eta\, \Pi\, g^{B}. \label{eq:bilateral}
\end{align}
LoRA initialization sets $B_e = 0$ at task start, so $P B_e = 0$ initially. The effective update factors as $\mathsf{R} = B_e^{+} A_e^{+} - B_e A_e = B_e\, \Delta A_e \,+\, \Delta B_e\, A_e \,+\, \Delta B_e\, \Delta A_e$, and protection is read off by evaluating $P\, \mathsf{R}\, h$ on $h \in \operatorname{range}(U)$.

The $B_e\, \Delta A_e\, h$ and $\Delta B_e\, \Delta A_e\, h$ terms both vanish by $\Delta A_e\, h = -\eta\, g^{A}\, (\Pi h) = 0$. Only $\Delta B_e\, A_e\, h$ survives, and its projection onto $U$ separates the two schemes:
\begin{align*}
P\, \mathsf{R}\, h \,\big|_{\text{one-sided}} \,&=\, -\eta\, \big( P\, g^{B} \big)\, A_e\, h \,=\, O(\eta), \\
P\, \mathsf{R}\, h \,\big|_{\text{bilateral}} \,&=\, -\eta\, \big( P\, \Pi \big)\, g^{B}\, A_e\, h \,=\, 0.
\end{align*}
The bilateral residual collapses through $P\, \Pi = 0$; the one-sided residual is first-order and accumulates as $-\eta \sum_{s \le S} P\, g^{B,s}$ over $S$ steps, while the bilateral invariant $P B_e = 0$ propagates inductively since $P\, \Delta B_e = -\eta\, (P\, \Pi)\, g^{B} = 0$. The Kronecker structure $\mathbf{P}_e^{t} = \Pi_e^{t} \otimes \Pi_e^{t}$ in the main text is therefore exact, not approximate.

\begin{remark}[Cross-task initial condition]
With LoRA's standard task-start reset $B_e \leftarrow 0$, $P_e^{t} B_e^{(0)} = 0$ holds and bilateral steps preserve the invariant. Without per-task reset, the increment-level guarantee $P_e^{t}\, \Delta(B_e A_e)\, h = 0$ for $h \in \mathrm{range}(U_e^{\leq t})$ still holds at every step: $B_e\, \Delta A_e$ and $\Delta B_e\, \Delta A_e$ vanish through $\Delta A_e\, h = 0$, and $P\, \Delta B_e = 0$ through $P\, \Pi = 0$, regardless of $B_e$'s initial value.
\end{remark}

\subsection{Proof of Proposition 3: FedAvg Preserves Orthogonality}
\label{supp:proof_prop3}

\begin{proof}[Proof of Proposition~3]
After the client-side bilateral projection of Section~\ref{supp:bilateral_residual}, each client gradient satisfies $P\, \tilde g_c^{\,B} = 0$ and $\tilde g_c^{\,A}\, P = 0$. Federated averaging with non-negative simplex weights $\{w_c\}$ collapses both factor sides through a single linearity step:
\begin{align*}
P\, \bar g^{\,B}
\,&=\, P \sum_{c=1}^{C} w_c\, \tilde g_c^{\,B}
\,\overset{(a)}{=}\, \sum_{c=1}^{C} w_c\, \big( P\, \tilde g_c^{\,B} \big)
\,=\, 0, \\
\bar g^{\,A}\, P
\,&=\, \bigg( \sum_{c=1}^{C} w_c\, \tilde g_c^{\,A} \bigg) P
\,\overset{(a)}{=}\, \sum_{c=1}^{C} w_c\, \big( \tilde g_c^{\,A}\, P \big)
\,=\, 0,
\end{align*}
where $(a)$ uses linearity of the projector. Feeding $P\,\bar g^{\,B} = 0$ and $\bar g^{\,A}\, P = 0$ into the residual expansion of Section~\ref{supp:bilateral_residual} gives $P\, \bar{\mathsf{R}} = 0 = \bar{\mathsf{R}}\, P$; the aggregate preserves $U$-orthogonality on both factors without server-side re-orthogonalization.
\end{proof}

\section{Extended Experiments}
\label{supp:experiments}

\subsection{External Sensitivity: Non-IID, Client Scale, and Warmup}
\label{supp:sensitivity}

\begin{figure*}[t]
\centering
\includegraphics[width=0.629\textwidth]{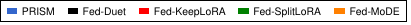}\\[-1.2em]
\subfloat[AA vs $\beta$, $C=5$]{%
    \includegraphics[width=0.245\textwidth]{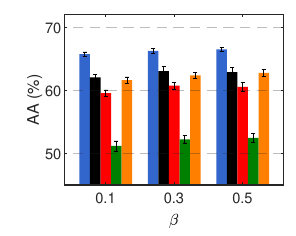}%
    \label{fig:sens:beta}%
}\hspace{0.005\textwidth}
\subfloat[AA vs $C$, $\beta=0.1$]{%
    \includegraphics[width=0.245\textwidth]{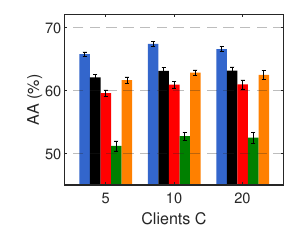}%
    \label{fig:sens:clients}%
}\hspace{0.005\textwidth}
\subfloat[AA vs warmup $s_0$]{%
    \includegraphics[width=0.245\textwidth]{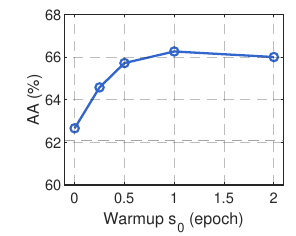}%
    \label{fig:sens:warmup}%
}
\caption{Hyperparameter and robustness sensitivity on LLaVA-1.5-7B. Panels~(a,~b): grouped bar charts across non-IID strength $\beta$ and client count $C$; \method{} dominates every configuration. Panel~(c): temporal warmup length $s_0$; \method{} is robust within a broad range around $s_0=1$ epoch, with the dashed grey line marking the no-warmup ($s_0=0$) performance floor.}
\label{fig:sensitivity}
\end{figure*}

Figure~\ref{fig:sensitivity} examines robustness to three external settings on LLaVA-1.5-7B. Panels~(a,~b) report \method{} against the strongest runner-up across non-IID concentration $\beta \in \{0.1, 0.3, 0.5\}$ and client count $C \in \{5, 10, 20\}$: \method{} holds the top position in every cell, and the \method{}-to-runner-up gap expands monotonically as $\beta$ decreases and $C$ increases; per-expert routing-weighted protection is most valuable precisely when the federation is least aligned. Panel~(c) confirms warmup robustness: AA stays within $\pm 0.4$\,pp for $s_0 \in [1, 2]$ epochs, well above the no-warmup floor (dashed grey). \method{} therefore does not rely on delicate tuning, and its headline results transfer across the realistic range of federated configurations.

\subsection{Cross-Architecture Generalization}
\label{supp:cross_arch}

This subsection supplies the full Qwen2.5-VL-7B~\cite{Bai2025ARXIV} evaluation on CoIN-6 and CoIN-Long-10 at $C=5$ across three non-IID concentrations $\beta \in \{0.1, 0.3, 0.5\}$, the ablation study on the same backbone, and qualitative evidence of forgetting dynamics and per-expert null-space consumption. All entries follow the sixteen-baseline protocol of main Tables~1 and~2; only the backbone is replaced.

\begin{table*}[t]
\centering
\caption{Cross-architecture generalization on Qwen2.5-VL-7B~\cite{Bai2025ARXIV}: main results on CoIN-6 at $C=5$. AA is in \%; BWT is in percentage points. \best{Red} is best and \second{blue} is second across competitive methods within each column. \textcolor{gray}{\textit{Italic gray}} rows are reference bounds.}
\label{tab:supp_qwen_coin6}
\renewcommand{\arraystretch}{1.15}
\setlength{\tabcolsep}{2.5pt}
\footnotesize
\begin{tabular}{l|cc|cc|cc}
\toprule
 & \multicolumn{2}{c|}{$\beta=0.1$} & \multicolumn{2}{c|}{$\beta=0.3$} & \multicolumn{2}{c}{$\beta=0.5$} \\
\cmidrule(lr){2-3} \cmidrule(lr){4-5} \cmidrule(lr){6-7}
Method & AA & BWT & AA & BWT & AA & BWT \\
\midrule
Zero-shot & {\color{gray}\textit{22.10}\std{0.32}} & {\color{gray}\textit{+0.02}\std{0.18}} & {\color{gray}\textit{22.18}\std{0.29}} & {\color{gray}\textit{+0.05}\std{0.22}} & {\color{gray}\textit{22.03}\std{0.30}} & {\color{gray}\textit{+0.08}\std{0.19}} \\
Multi-task & {\color{gray}\textit{69.31}\std{0.34}} & {\color{gray}\textit{-0.05}\std{0.20}} & {\color{gray}\textit{69.54}\std{0.31}} & {\color{gray}\textit{+0.12}\std{0.18}} & {\color{gray}\textit{69.42}\std{0.28}} & {\color{gray}\textit{-0.02}\std{0.21}} \\
\cmidrule(lr){1-7}
FedProx~\cite{Li2020MLSys} & 47.80\std{0.52} & -5.80\std{0.74} & 52.60\std{0.44} & -5.20\std{0.76} & 51.70\std{0.48} & -5.45\std{0.82} \\
Fed-EWC~\cite{Kirkpatrick2017PNAS} & 49.10\std{0.65} & -7.20\std{0.68} & 46.70\std{0.58} & -7.50\std{0.61} & 44.20\std{0.71} & -8.20\std{0.84} \\
Fed-LwF~\cite{Li2016ECCV} & 49.50\std{0.53} & -6.30\std{0.79} & 51.80\std{0.49} & -6.60\std{0.83} & 47.70\std{0.64} & -6.90\std{0.86} \\
Fed-Replay~\cite{Rebuffi2017CVPR} & 50.02\std{0.59} & -7.80\std{0.72} & 52.34\std{0.51} & -8.10\std{0.70} & 51.85\std{0.63} & -8.40\std{0.74} \\
\cmidrule(lr){1-7}
Fed-GPM~\cite{Saha2021ICLR} & 47.45\std{0.70} & -9.45\std{0.88} & 48.90\std{0.67} & -9.80\std{0.92} & 49.55\std{0.61} & -10.15\std{0.79} \\
Fed-O-LoRA~\cite{Wang2023EMNLP} & 43.88\std{0.79} & -10.48\std{0.80} & 45.12\std{0.73} & -10.85\std{0.84} & 45.76\std{0.68} & -11.25\std{0.92} \\
Fed-KeepLoRA~\cite{Luo2026ICLR} & 53.12\std{0.46} & -4.78\std{0.75} & 54.46\std{0.42} & -5.12\std{0.78} & 54.84\std{0.55} & -5.55\std{0.83} \\
Fed-SplitLoRA~\cite{Qiu2026ICLR} & 45.82\std{0.72} & -8.56\std{0.84} & 47.08\std{0.66} & -8.94\std{0.89} & 47.70\std{0.58} & -9.32\std{0.75} \\
\cmidrule(lr){1-7}
Fed-MoELoRA~\cite{Chen2024NIPS} & 49.32\std{0.61} & -7.65\std{0.66} & 50.44\std{0.53} & -8.05\std{0.72} & 50.88\std{0.79} & -8.48\std{0.88} \\
Fed-SMoLoRA~\cite{Wang2025ICCV} & 50.46\std{0.54} & -6.95\std{0.61} & 51.70\std{0.48} & -7.32\std{0.65} & 52.12\std{0.62} & -7.75\std{0.79} \\
Fed-MoDE~\cite{Wei2025NIPS} & 54.80\std{0.55} & -9.60\std{0.76} & 55.82\std{0.50} & -10.06\std{0.81} & 56.04\std{0.47} & -10.52\std{0.85} \\
Fed-PCLR~\cite{Meng2026ICLR} & 46.34\std{0.63} & \second{-3.08}\std{0.71} & 47.55\std{0.57} & \second{-3.32}\std{0.60} & 48.22\std{0.68} & \second{-3.76}\std{0.83} \\
Fed-EProj~\cite{He2026TIP} & 50.20\std{0.52} & -12.10\std{0.88} & 51.43\std{0.45} & -12.47\std{0.89} & 51.90\std{0.50} & -12.95\std{0.93} \\
Fed-Duet~\cite{Guo2026ICLR} & \second{55.48}\std{0.68} & -5.92\std{0.74} & \second{56.64}\std{0.62} & -6.24\std{0.72} & \second{57.16}\std{0.54} & -6.56\std{0.78} \\
\cmidrule(lr){1-7}
\rowcolor{prismblue!12}
\method{} (Ours) & \best{59.80}\std{0.42} & \best{+4.50}\std{0.48} & \best{58.30}\std{0.39} & \best{+2.70}\std{0.45} & \best{58.90}\std{0.45} & \best{+0.30}\std{0.51} \\
\bottomrule
\end{tabular}
\end{table*}

\begin{table*}[t]
\centering
\caption{Cross-architecture generalization on Qwen2.5-VL-7B: main results on CoIN-Long-10 at $C=5$. \best{Red} is best and \second{blue} is second across competitive methods within each column. \textcolor{gray}{\textit{Italic gray}} rows are reference bounds.}
\label{tab:supp_qwen_long10}
\renewcommand{\arraystretch}{1.15}
\setlength{\tabcolsep}{2.5pt}
\footnotesize
\begin{tabular}{l|cc|cc|cc}
\toprule
 & \multicolumn{2}{c|}{$\beta=0.1$} & \multicolumn{2}{c|}{$\beta=0.3$} & \multicolumn{2}{c}{$\beta=0.5$} \\
\cmidrule(lr){2-3} \cmidrule(lr){4-5} \cmidrule(lr){6-7}
Method & AA & BWT & AA & BWT & AA & BWT \\
\midrule
Zero-shot & {\color{gray}\textit{19.35}\std{0.30}} & {\color{gray}\textit{+0.05}\std{0.21}} & {\color{gray}\textit{19.40}\std{0.32}} & {\color{gray}\textit{+0.08}\std{0.19}} & {\color{gray}\textit{19.48}\std{0.28}} & {\color{gray}\textit{+0.02}\std{0.22}} \\
Multi-task & {\color{gray}\textit{49.85}\std{0.35}} & {\color{gray}\textit{-0.15}\std{0.24}} & {\color{gray}\textit{50.12}\std{0.31}} & {\color{gray}\textit{-0.12}\std{0.23}} & {\color{gray}\textit{50.02}\std{0.30}} & {\color{gray}\textit{-0.08}\std{0.20}} \\
\cmidrule(lr){1-7}
FedProx~\cite{Li2020MLSys} & 26.68\std{0.58} & -20.48\std{0.88} & 27.15\std{0.53} & -20.65\std{0.84} & 27.58\std{0.60} & -20.82\std{0.91} \\
Fed-EWC~\cite{Kirkpatrick2017PNAS} & 22.08\std{0.62} & -12.55\std{0.75} & 22.45\std{0.57} & -12.80\std{0.72} & 22.85\std{0.64} & -13.10\std{0.81} \\
Fed-LwF~\cite{Li2016ECCV} & 24.32\std{0.64} & -18.85\std{0.90} & 24.85\std{0.58} & -19.20\std{0.86} & 25.40\std{0.66} & -19.58\std{0.93} \\
Fed-Replay~\cite{Rebuffi2017CVPR} & 32.85\std{0.68} & -12.10\std{0.80} & 33.42\std{0.62} & -12.45\std{0.78} & 33.88\std{0.59} & -12.78\std{0.83} \\
\cmidrule(lr){1-7}
Fed-GPM~\cite{Saha2021ICLR} & 23.08\std{0.75} & -18.28\std{0.94} & 23.60\std{0.70} & -18.65\std{0.92} & 24.05\std{0.65} & -18.95\std{0.85} \\
Fed-O-LoRA~\cite{Wang2023EMNLP} & 17.82\std{0.82} & -18.45\std{0.86} & 18.25\std{0.77} & -18.85\std{0.88} & 18.62\std{0.72} & -19.20\std{0.95} \\
Fed-KeepLoRA~\cite{Luo2026ICLR} & 32.42\std{0.51} & -10.08\std{0.83} & 32.95\std{0.47} & -10.48\std{0.82} & 33.35\std{0.60} & -10.82\std{0.88} \\
Fed-SplitLoRA~\cite{Qiu2026ICLR} & 19.42\std{0.76} & -15.45\std{0.88} & 19.95\std{0.69} & -15.85\std{0.92} & 20.38\std{0.62} & -16.22\std{0.78} \\
\cmidrule(lr){1-7}
Fed-MoELoRA~\cite{Chen2024NIPS} & 24.58\std{0.65} & -14.82\std{0.72} & 25.10\std{0.57} & -15.25\std{0.76} & 25.55\std{0.82} & -15.62\std{0.92} \\
Fed-SMoLoRA~\cite{Wang2025ICCV} & 28.38\std{0.57} & -11.45\std{0.68} & 28.90\std{0.52} & -11.85\std{0.72} & 29.32\std{0.66} & -12.22\std{0.83} \\
Fed-MoDE~\cite{Wei2025NIPS} & 30.98\std{0.58} & -19.35\std{0.82} & 31.55\std{0.54} & -19.78\std{0.86} & 32.02\std{0.51} & -20.18\std{0.90} \\
Fed-PCLR~\cite{Meng2026ICLR} & 19.15\std{0.66} & \second{-7.45}\std{0.75} & 19.65\std{0.60} & \second{-7.85}\std{0.64} & 20.08\std{0.71} & \second{-8.22}\std{0.87} \\
Fed-EProj~\cite{He2026TIP} & 25.42\std{0.55} & -21.42\std{0.92} & 25.92\std{0.48} & -21.85\std{0.94} & 26.35\std{0.53} & -22.22\std{0.97} \\
Fed-Duet~\cite{Guo2026ICLR} & \second{34.62}\std{0.72} & -10.02\std{0.78} & \second{35.20}\std{0.65} & -10.35\std{0.76} & \second{35.68}\std{0.58} & -10.68\std{0.82} \\
\cmidrule(lr){1-7}
\rowcolor{prismblue!12}
\method{} (Ours) & \best{40.55}\std{0.45} & \best{+2.50}\std{0.52} & \best{39.48}\std{0.42} & \best{+1.38}\std{0.48} & \best{40.22}\std{0.48} & \best{+0.45}\std{0.55} \\
\bottomrule
\end{tabular}
\end{table*}

Tables~\ref{tab:supp_qwen_coin6} and~\ref{tab:supp_qwen_long10} reproduce the LLaVA-1.5 ranking pattern of main Tables~1 and~2 on Qwen2.5-VL-7B: \method{} holds the top rank on both AA and BWT across all six $(\beta, \text{benchmark})$ cells, with Fed-Duet and Fed-PCLR remaining the AA and BWT runners-up; \method{} is the only method retaining positive backward transfer across all three $\beta$ settings on CoIN-Long-10. Absolute AA on Qwen is lower than LLaVA-1.5-7B by roughly six to eight percentage points, likely reflecting backbone-specific generation patterns under exact-match scoring rather than any methodological difference; ranking within each paradigm group is preserved on every cell, confirming that C1--C5 transfer beyond the LLaVA family.

\begin{table}[t]
\centering
\caption{Component ablation of \method{} on Qwen2.5-VL-7B at $\beta=0.3$, $C=5$, CoIN-6. Each row removes a single mechanism. Component ordering matches main Table~3 on LLaVA-1.5.}
\label{tab:supp_qwen_ablation}
\renewcommand{\arraystretch}{1.2}
\setlength{\tabcolsep}{4pt}
\scriptsize
\begin{tabular}{l|cc}
\toprule
Variant & AA & BWT \\
\midrule
\rowcolor{prismblue!18}
\method{} (full) & \best{58.30\std{0.39}} & \best{+2.70\std{0.45}} \\
\midrule
\rowcolor{prismblue!5}
w/o PE-FOSU & 53.72\std{0.55} & -1.45\std{0.72} \\
w/o per-expert basis & 54.68\std{0.51} & -0.18\std{0.59} \\
\rowcolor{prismblue!5}
w/o routing-weighted & 55.12\std{0.47} & +0.46\std{0.63} \\
w/o router freeze & 54.25\std{0.58} & -0.08\std{0.66} \\
\rowcolor{prismblue!5}
w/o scheduling & 56.28\std{0.43} & +1.75\std{0.54} \\
\bottomrule
\end{tabular}
\end{table}

Table~\ref{tab:supp_qwen_ablation} reproduces the main-text ablation on Qwen2.5-VL-7B for five of the seven components. The ordering matches LLaVA-1.5 within a few tenths of a percentage point: PE-FOSU removal is sharpest ($-4.58$\,pp), then router freezing ($-4.05$), per-expert basis ($-3.62$), routing-weighted accumulation ($-3.18$), and scheduling ($-2.02$). Forgetting trajectories and per-expert null-space consumption on Qwen (plots omitted for space) preserve the qualitative signatures of main Figs.~3 and~5.

\subsection{Implementation and Per-Layer Budget Schedule}
\label{supp:implementation}

Table~\ref{tab:supp_budget} reports the per-layer protection budget that the water-filling rule $k_l^\star = \bar k\, \gamma_l^2 / \|\boldsymbol{\gamma}\|_2^2$ outputs once the interference landscape $\gamma_l$ is measured. For compact reporting, adjacent MoE-LoRA layers with similar $\gamma_l$ are grouped into three $\gamma$-ranked tiers and the table lists the per-tier average; the deployed budgets are the per-layer water-filling outputs themselves, requiring no tuning beyond the one-off $\gamma_l$ measurement. Per-step projection overhead is below $3\%$ FLOPs across all backbones; the end-of-task PE-FOSU union, executed once per task as a thin SVD on a $d \times (C+1)k$ concatenation, is negligible relative to training.

\begin{table}[t]
\centering
\caption{Per-layer protection budget schedule produced by the $\gamma$-driven water-filling rule. MoE-LoRA layers on each backbone are split into three tiers; $(p, k/r_o)$ follows $\gamma$-rank without additional tuning.}
\label{tab:supp_budget}
\renewcommand{\arraystretch}{1.15}
\setlength{\tabcolsep}{4pt}
\footnotesize
\begin{tabular}{l|c|c|c|c}
\toprule
Backbone & Layers & $\gamma_l$ & $p$ & $k/r_o$ \\
\midrule
\multirow{3}{*}{Qwen2.5-VL-7B} & L20--21 & 0.62 & 0.95 & 0.35 \\
 & L22--25 & 0.52 & 0.90 & 0.25 \\
 & L26--27 & 0.40 & 0.85 & 0.15 \\
\midrule
\multirow{3}{*}{LLaVA-1.5-7B} & L24--25 & 0.60 & 0.95 & 0.35 \\
 & L26--29 & 0.51 & 0.90 & 0.25 \\
 & L30--31 & 0.39 & 0.85 & 0.15 \\
\midrule
\multirow{3}{*}{LLaVA-1.5-13B} & L32--33 & 0.64 & 0.95 & 0.35 \\
 & L34--37 & 0.53 & 0.90 & 0.25 \\
 & L38--39 & 0.41 & 0.85 & 0.15 \\
\bottomrule
\end{tabular}
\end{table}

\subsection{Task-Expert Routing Heatmaps and Routing Stability}
\label{supp:routing_stability}

\begin{figure}[t]
\centering
\subfloat[LLaVA-1.5-7B.]{%
    \includegraphics[width=0.48\linewidth]{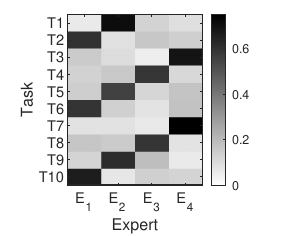}%
    \label{fig:supp:hm7b}%
}\hfill
\subfloat[LLaVA-1.5-13B.]{%
    \includegraphics[width=0.48\linewidth]{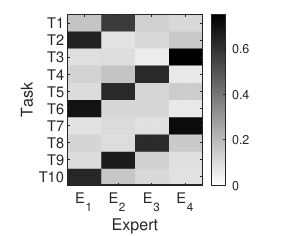}%
    \label{fig:supp:hm13b}%
}
\caption{Task-expert routing weight heatmap on CoIN-Long-10 after the router freeze at task 1. Each task preferentially activates one or two experts, and the assignment remains stable across the whole sequence; this visually complements the null-space consumption of main Fig.~5.}
\label{fig:supp_heatmap}
\end{figure}

Figure~\ref{fig:supp_heatmap} reports the task-expert routing weight heatmap on CoIN-Long-10 after the router freeze at task 1. Because \method{} freezes the router after $T_1$, top-1 flip rate and weight drift are zero by construction; routing entropy remains near zero throughout the ten-task sequence. Baselines with trainable routers (Fed-MoDE, Fed-SMoLoRA, Fed-MoELoRA) instead show continuous drift in all three metrics, consistent with the view that routing instability is a design-level property of trainable routers rather than a transient training artifact.


\end{document}